\documentclass[11pt]{article}
\usepackage[margin=1in]{geometry}
\usepackage{amsmath,amssymb,amsthm,booktabs,graphicx,hyperref,url}
\usepackage[T1]{fontenc}
\makeatletter
\g@addto@macro\ttfamily{\hyphenchar\font=`\-\relax}
\makeatother
\makeatletter
\renewcommand\section{\@startsection{section}{1}{\z@}%
  {-3.5ex \@plus -1ex \@minus -.2ex}{2.3ex \@plus.2ex}%
  {\raggedright\normalfont\Large\bfseries}}
\renewcommand\subsection{\@startsection{subsection}{2}{\z@}%
  {-3.25ex\@plus -1ex \@minus -.2ex}{1.5ex \@plus .2ex}%
  {\raggedright\normalfont\large\bfseries}}
\renewcommand\subsubsection{\@startsection{subsubsection}{3}{\z@}%
  {-3.25ex\@plus -1ex \@minus -.2ex}{1.5ex \@plus .2ex}%
  {\raggedright\normalfont\normalsize\bfseries}}
\makeatother
\newtheorem{theorem}{Theorem}
\newtheorem{lemma}{Lemma}
\newtheorem{proposition}{Proposition}

\title{WitCert: Sound Runtime Risk Observability\\ and Gating for KV-Cache Quantization}
\author{%
  Fanzhe Wei\\ Metask Lab\\ {\small\texttt{whyer1@gmail.com}}
  \and Li Liu\\ Metask Lab\\ {\small\texttt{muriel092611@gmail.com}}
  \and Ziyang Wang\\
  {\small\begin{tabular}[t]{@{}c@{}}
   Logistics Science and Technology Innovation Integrated\\
   Development Research Center of Shaanxi Logistics Group,\\
   Xi'an Jiaotong University, Xi'an, China\\
   and\\
   School of Management,\\
   Xi'an Jiaotong University, Xi'an, China\\[2pt]
   \texttt{wang-zy13@tsinghua.org.cn}
   \end{tabular}}
  \and Chenyu Wang\\ Metask Lab\\ {\small\texttt{cnados@gmail.com}}}
\date{}

\begin{document}

\maketitle

\begin{abstract}
KV-cache quantization is validated today by offline benchmark averages; a deployed system cannot tell whether compression is damaging the request it is serving right now. We give it a provably sound runtime meter---a ``DTrace for KV quantization'': a per-(layer, head, step) upper bound on the total variation between exact and compressed attention. The meter has two tiers: a deterministic band-norm-witness bound, sound for \emph{any} cache-preserving black-box quantizer and for \emph{any} query (adaptive-safe, worst-case Cauchy--Schwarz plus RoPE band-unitarity), and a tighter probabilistic certificate for a controlled subtractively-dithered INT8 quantizer under an explicit request-level failure budget (stated for non-adaptive queries; core theorems machine-checked in Lean~4). Three results. \emph{Observability}: the meter enters SGLang through an env-guarded patch, and any scheme registered as one tensor function is measured in live serving. \emph{Repair}: meter-driven gating---risk-ranked where the witness is saturated, certified where it is informative---empirically restores the quality floor at benchmark scale, e.g.\ raw-cast fp8 from 22.8 back to 79.7 on hard RULER tasks with the difference from uncompressed bounded at $[+0.0,+0.8]$ by a paired test. \emph{Analysis}: aggressive schemes survive on cross-layer error cancellation, not per-step fidelity---in a 28-layer sweep, no single layer's pollution alone loses anything (0/28)---and the certified int8 cache serves $1.88\times$ more KV tokens at the same memory in SGLang.
All artifacts, guards, and the Lean development are released at
\url{https://github.com/metask-ai/witcert-kv-certificates}; every number
regenerates from the shipped artifacts by one command.
\end{abstract}

\section{Introduction}
\label{sec:intro}
Autoregressive decoding requires caching keys and values for every past token; beyond 32k context a 7B-class model spends gigabytes on the KV cache, and at 128k the cache exceeds the weights themselves. (Our quality evaluation spans 4--32k contexts; the memory account is measured to 128k; see the scope statement of Sec.~\ref{sec:regime}.) The four established compression routes---token eviction (H2O~\cite{h2o2023}, SnapKV~\cite{snapkv2024}), channel/frequency pruning (ThinK~\cite{think2024}, RAP~\cite{rap2026}), low-rank factorization (Palu~\cite{palu2025}, KQ-SVD~\cite{kqsvd2025}), and quantization (KIVI~\cite{kivi2024}, KVQuant~\cite{kvquant2024})---share one defect: they are \emph{open loop}; dynamic sparsity (Quest~\cite{quest2024}) selects at run time but is likewise validated only offline. The policy is fixed offline or heuristically, and the error actually incurred on the current input is neither observable nor controllable. We observed the extreme form of this defect: under a query-agnostic protocol (the regime studied by the query-agnostic compression line~\cite{kvzip2025}), SnapKV drops from 100 to 1.1--65 on RULER~\cite{ruler2024} needle tasks while the serving system emits no signal whatsoever (Sec.~\ref{sec:ruler}).

Turning compression from open loop into closed loop requires a \emph{runtime-computable} error bound. Runtime-Certified Quantized Attention~\cite{runtime_certified_2026} first supplied such certificates, but its bound is worst-case ($\mathrm{TV} \le \tanh(\max|\varepsilon|)$, data independent): we measured that at every practical threshold ($\tau \le 0.2$) it authorizes \emph{zero} compression---the certificate holds, but is too conservative to be useful.

Our starting point is therefore: \textbf{for a certificate to authorize compression, it must be data dependent}. Two obstacles follow. First, a data-dependent bound needs the per-token error magnitude, yet the error depends on a query that has not yet arrived. Second, modern caches store keys \emph{after} the RoPE rotation~\cite{su2021roformer}, so any witness must be invariant to position or it must be recomputed per position. Our answer is the residual witness (Sec.~\ref{sec:tiera}): the band-wise norms of the quantization residual solve both at once---they are query independent, computable once at write time, and position invariant because RoPE acts unitarily within a frequency band (Lemma~\ref{lem:rope}).

\textbf{Terminology.} We call the runtime quantity the \emph{meter}: a sound upper bound on the attention total variation, reported as $\min(1,\cdot)$ since $\mathrm{TV}\le 1$ always. The stored per-token summary it is computed from is the \emph{witness}. When the meter is below saturation it is a \emph{certificate}---a mathematically valid guarantee; above saturation its raw value remains an empirically discriminative \emph{risk score} but guarantees nothing. Gating experiments are labelled \emph{certified} ($\tau<1$) or \emph{risk-ranked} ($\tau\ge 1$) accordingly.

Contributions: (1) a sound, position-invariant \emph{universal meter}: the band-norm witness theorem for black-box \emph{cache-preserving} schemes (Sec.~\ref{sec:tiera}), plus a cautionary counterexample, its sound replacement bound, saturation-aware $\min(1,\cdot)$ reporting, and a two-sided companion bound whose measured non-benefit we report (Secs.~\ref{sec:counterexample}, \ref{sec:twosided}); (2) a uniform sub-Gaussian runtime certificate for a controlled quantizer, with an explicit request-level failure budget and \emph{no} witness storage, plus two negative controls on bound selection (Sec.~\ref{sec:subg}); (3) a mechanism finding that benchmarks cannot produce: the observable failure of aggressive schemes is \emph{cross-layer accumulation}, not per-step infidelity---in a 28-layer sweep no single layer's pollution alone loses anything (0/28)---which explains both why aggressive schemes survive and why per-step bounds are intrinsically conservative (Sec.~\ref{sec:mechanism}); (4) a system implementation in which the meter is structurally cheap and CUDA-graph capturable; the certify--fallback loop is closed and validated in our standalone harness (head- and block-granular paging, Secs.~\ref{sec:closedloop}, \ref{sec:blockpage}) and, in live SGLang~\cite{sglang2024} serving, through the observatory pool at benchmark scale (Sec.~\ref{sec:gatedeval}); the \emph{packed} certified cache in SGLang has certificate computation, storage, telemetry and fallback integrated, and its gated end-to-end quality is validated in Sec.~\ref{sec:h200}; the remaining systems cost is repeated page-in without a request-level repair cache.

\section{Related Work}
\textbf{Low-rank.} Palu~\cite{palu2025} and STAR-KV~\cite{starkv2026} both compress the pre-RoPE representation and reconstruct at decode time, because RoPE blocks weight absorption on the key side (Palu: ``the non-linear nature of these positional embeddings prevents the matrix fusion''). Palu itself prices this cost: $2.91\times$ attention speedup with RoPE versus $6.17\times$ without. KQ-SVD~\cite{kqsvd2025} gives a closed-form optimal low-rank factorization with a global Lipschitz bound on the attention output, but has no margin or per-row analysis and does not discuss RoPE at all. \emph{No prior work performs certified compression of post-RoPE keys.}

\textbf{Sketching.} BalanceKV~\cite{balancekv2025} provides an end-to-end attention-output guarantee with a matching lower bound, but its constant $e^{2r^2/\sqrt{d}}$ is roughly $e^{39.8}$ at measured norms, and RoPE is never mentioned---the theory implicitly treats keys as position independent. Compactor~\cite{compactor2025} bounds the spectrum of the \emph{pre}-RoPE $K^\top K$ and concedes that its experimental parameters violate its own theorem's conditions.

\textbf{RoPE-aware compression.} RAP~\cite{rap2026} prunes RoPE column pairs and proves commutation with rotation---the rigidity proposition of Appendix~\ref{app:rigidity} supplies the converse completeness---but requires LoRA distillation to recover. EliteKV~\cite{elitekv2025} and KV-Latent~\cite{kvlatent2025} perform frequency selection or resampling, both requiring retraining and offering no bound. Block-GTQ~\cite{blockgtq2026} is the closest neighbour on the quantization axis (RoPE-aware bit allocation, real packed path) but provides no runtime certificate.

\textbf{Certified compression.} Runtime-Certified Quantized Attention~\cite{runtime_certified_2026} is our primary interlocutor; we borrow its two-term decomposition and telemetry framing. The difference is that its softmax bound is a worst-case $\tanh(\Delta)$ with $\Delta$ determined by the quantization step and query norm, independent of the per-token error actually incurred. We compare both against that \emph{propagation bound} (verified sound by our own adversarial testing) under identical witnesses and against a full-system reproduction from the paper specification (Sec.~\ref{sec:h200}). The latter includes block-level precision escalation, INT4 values, CPU backing storage and a two-pass scan; because no public implementation is available, we declare every reproduction deviation and report fidelity checks. The bound-shape comparison is deliberately two-sided: the det-tanh gate is cheaper in the loose-witness regime, whereas ours is tighter for the controlled-quantizer regime. On the competitive map: RaBitQCache~\cite{rabitqcache2026} (binary-proxy score with a proven error bound, no TV certificate), Block-GTQ~\cite{blockgtq2026} (RoPE bit allocation, real packed path, no runtime certificate), KVQuant~\cite{kvquant2024} (dense-and-sparse per-vector outliers, pre-RoPE), and KVarN~\cite{kvarn2026}/CommVQ~\cite{commvq2025} (extreme low-bit baselines); the minimax analysis of~\cite{riskkv2026} likewise argues for budgeting compression against the observed query distribution. A rigidity proposition (stated and proved in Appendix~\ref{app:rigidity}) characterizes static RoPE-equivariant linear compression as exactly frequency-pair selection with per-frequency scalars; it explains the field's current split---the RoPE-structure-aware methods carry no guarantee, and the methods that carry a guarantee route around RoPE---and justifies handling outliers at \emph{pair} granularity. A companion paper~\cite{witprobe-companion} generalizes this line beyond the single KV cache: typed per-stage observability contracts over heterogeneous attention memory (latent caches, learned sparse selectors, pooled hierarchies, recurrent state), with the present meter serving as the storage-stage instance.

\textbf{Conformal prediction and distribution-free UQ.} Our request-level failure budget $\delta_{\mathrm{req}}$ and its union-bound allocation speak the language of distribution-free uncertainty quantification~\cite{vovk2005,angelopoulos2021}, so a natural question is why we do not simply calibrate a split-conformal quantile of the true TV. The two guarantees are not interchangeable. Split conformal yields \emph{marginal} coverage---averaged over requests exchangeable with a calibration set---whereas the meter is a \emph{per-request, conditional} statement: it holds for the request being served, with no calibration set and no exchangeability assumption. The distinction matters precisely because compression fails out of distribution: the silent needle collapses of Sec.~\ref{sec:ruler} occur exactly where a calibration set drawn from typical traffic stops being representative. Conformal calibration would also need the true TV as its conformity score, i.e.\ one uncompressed forward pass per calibration point---the cost the runtime meter avoids. The two are complementary rather than competing: calibrating the gate threshold $\tau$ conformally on live traffic is attractive future work.

\section{A Plausible but Unsound Certificate, and its Sound Replacement}
\label{sec:counterexample}
\textbf{Setup and notation.} Fix one (layer, head, decode step). Let $p$ and $\tilde p$ be the softmax attention distributions over the $S$ cached tokens computed from exact and compressed keys respectively, $\varepsilon_t$ the logit error of token $t$, and $c_t \ge |\varepsilon_t|$ any per-token error bound. We meter the total variation $\mathrm{TV}(p,\tilde p) = \tfrac12 \sum_t |p_t - \tilde p_t|$. TV is the right target for two reasons: it is the operational quantity---the amount of attention mass that moves, so a silently dropped needle token is by definition a TV event---and it propagates linearly to the attention output: $\|y-\hat y\| \le 2\,\mathrm{TV}\max_t\|v_t\| + \sum_t \tilde p_t \rho^v_t$, where $\rho^v_t := \|\hat v_t - v_t\|$ is token $t$'s value reconstruction error (Sec.~\ref{sec:subg}), so a TV budget is an output-error budget. All norms are Euclidean; we write $\mathrm{TV}_K$ when the value side is also in scope. Since $\mathrm{TV} \le 1$ holds trivially, any upper bound is informative only below $1$; the meter is therefore reported as $\min(1,\cdot)$ throughout.

\textbf{The flawed certificate.} With $\bar c = \mathbb{E}_{\tilde p}[c]$ and $c_{\max} = \max_t c_t$, a natural-looking bound is
\begin{equation}
\mathrm{TV}(p,\tilde p) \;\overset{?}{\le}\; \tfrac12 e^{2c_{\max}}\bigl(\mathbb{E}_{\tilde p}\bigl[|c-\bar c|\bigr] + \bar c\,(e^{c_{\max}}-1)\bigr),
\label{eq:flawed}
\end{equation}
whose intuition is ``errors are attention weighted, so a small \emph{dispersion of the bound} implies small TV.'' This is wrong: $c_t$ constrains only $|\varepsilon_t|$. When all $c_t$ are equal the first-order term vanishes, yet errors of opposite sign still produce an $O(c)$ change in TV.

\textbf{Minimal counterexample.} Take $\tilde p = (\tfrac12,\tfrac12)$, $c = (0.1,0.1)$, $\varepsilon = (0.1,-0.1)$. The true total variation is $0.049834$, while~\eqref{eq:flawed} returns $0.006423$---a factor of $7.8$. Notably, \eqref{eq:flawed} passed $32{,}256$ checks in our own earlier experiments because those checks compared \emph{mean} TV against \emph{mean} bound over a batch of queries; per-query checking exposes the violation immediately. We report this as a methodological lesson: certificates must be validated per query and adversarially, never in aggregate.

\begin{theorem}[Sound replacement]\label{thm:eform}
Let $A = \mathbb{E}_{\tilde p}[e^{c}]$. If $|\varepsilon_t| \le c_t$ for all $t$, then
$\mathrm{TV}(p,\tilde p) \le \tfrac12\,(A^2-1)$.
\end{theorem}
\begin{proof}
Since $|\varepsilon_t| \le c_t$, Cauchy--Schwarz gives $\mathbb{E}_{\tilde p}[e^{c}]\,\mathbb{E}_{\tilde p}[e^{-c}] \ge 1$ and hence $1/A \le \mathbb{E}_{\tilde p}[e^{-\varepsilon}] \le A$, so
$\bigl|\log(p_t/\tilde p_t)\bigr| \le c_t + \log A$ and therefore $|p_t/\tilde p_t - 1| \le A e^{c_t} - 1$.
Taking expectation under $\tilde p$ gives $\mathrm{TV} = \tfrac12\mathbb{E}_{\tilde p}\bigl[|p/\tilde p - 1|\bigr] \le \tfrac12 (A^2-1)$.
\end{proof}
The first-order behaviour is $\bar c$, the attention-weighted \emph{mean} of the bound. We validated the replacement with 500{,}000 adversarial trials (extreme sign patterns, peaked and flat distributions, large $c$) and froze both the counterexample and the adversarial suite as a regression guard in the repository. Numerically we evaluate $A$ in the log domain and return the trivial bound $1$ on overflow; we never truncate the error parameter itself.

\section{The Two-Tier Meter: a Universal Witness Bound and a Sub-Gaussian Certificate}
\label{sec:subg}

\subsection{Tier A: a deterministic witness bound for
  black-box cache-preserving schemes}
\label{sec:tiera}
Modern caches store keys \emph{after} the RoPE rotation, so any stored error summary must be invariant to the (not yet known) query position. Write $d$ for the head dimension and group the $d/2$ RoPE frequency pairs into $B$ contiguous bands (the system integration uses $B=16$, i.e.\ 32 B/tok/head; the offline six-quantizer study of Fig.~\ref{fig:f5} uses $B=8$, 16 B; band $b$ spans the coordinates of $d/(2B)$ frequency pairs). \textbf{Scope.} Tier A covers any scheme that, for every cached token, materializes a reconstructed key $\hat k_t$ of the same shape---i.e.\ arbitrary quantizers and any transform whose output can be written back as a per-token key. Eviction, token pooling/merging, latent caches and recurrent-state compressions store \emph{no} per-token reconstruction, so no residual exists and Tier A does not apply to them; certifying eviction online is exactly open problem (3) of Sec.~\ref{sec:open}. Within this scope, at write time let $r_t = \hat k_t - k_t$ be the residual, and store the witness
\begin{equation}\label{eq:witness}
w_{t,b} \;=\; \|r_{t,b}\|, \qquad b = 1,\dots,B,
\end{equation}
the per-band Euclidean norms of the residual (stored in fp16, $2B = 32$ bytes per token per head, written B/tok/head below).

\begin{lemma}[RoPE band unitarity]\label{lem:rope}
Let $R_n$ be the RoPE rotation at relative position $n$ (we reserve $m$ for the outlier-pair count): block diagonal over frequency pairs, acting on pair $j$ as the $2\times2$ rotation by angle $n\theta_j$. Every frequency pair lies inside a single band, so for every band $b$, every $x$ and every $n$, $\|(R_n x)_b\| = \|x_b\|$. In particular the witness of a post-RoPE residual equals the witness of the corresponding pre-RoPE residual, and is invariant to the query position.
\end{lemma}
\begin{proof}
$R_n$ restricted to band $b$ is a direct sum of $2\times2$ rotations, hence orthogonal on the band's subspace.
\end{proof}

\begin{theorem}[Sound black-box logit bound]\label{thm:tiera}
For any query $q$ (post-RoPE, softmax scale $1/\sqrt d$) and any cached token $t$ with residual $r_t$,
\[
|\varepsilon_t| \;=\; \frac{|q\cdot r_t|}{\sqrt d} \;\le\; \frac{1}{\sqrt d}\sum_{b=1}^{B} \|q_b\|\,w_{t,b} \;=:\; c_t,
\]
by Cauchy--Schwarz applied within each band together with Lemma~\ref{lem:rope}. The bound is computable at decode time from the current query and the stored witness alone, for any cache-preserving scheme, with no assumption on the residual distribution; feeding $c_t$ into Theorem~\ref{thm:eform} yields a sound TV meter for any such scheme.
\end{theorem}

Three remarks. (i) \emph{Proof--kernel correspondence}: the deployed kernel computes exactly this quantity---\texttt{witness\_of} stores the witness~\eqref{eq:witness} at write time and the decode kernel accumulates $c_t = \mathrm{sm\_scale}\cdot\sum_b \|q_b\| w_{t,b}$---so the theorem and the implementation match term by term. (ii) \emph{Cost}: 32 B/tok/head of storage ($12.5\%$ of an fp16 key) and one length-$B$ dot product per (query, token), which shares the attention kernel's data pass (Sec.~\ref{sec:sglang}). (iii) \emph{Tightness}: the bound is tight when the residual concentrates in few bands and is loose when residuals are dense and sign-alternating---$1289\times$ loose at 4 bits---which is why the universal meter saturates for 1--2 bit schemes: \emph{with the current band-norm witness}, per-step certification of such schemes is vacuous---a statement about this witness, not an impossibility proof; whether a better constant-size witness exists below 8 bits is precisely open problem (1) of Sec.~\ref{sec:open}. (The cross-layer mechanism of Sec.~\ref{sec:mechanism} explains why any per-step bound, however tight, must be conservative for aggressive schemes.)

\begin{figure}[t]\centering
\includegraphics[width=0.92\linewidth]{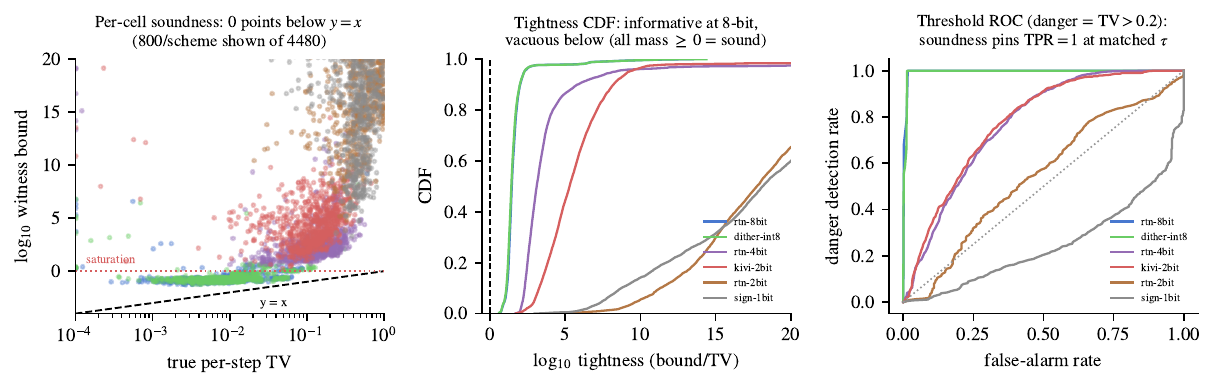}
\caption{The universal meter on six black-box quantizers, per-cell diagnostics ($n{=}4480$ cells each, unclamped log-domain values). \emph{Left}: per-cell scatter---zero points below $y{=}x$ (soundness); below the saturation line the meter is also informative. \emph{Middle}: tightness CDF (bound/TV)---8-bit schemes concentrate near $10^{1.5}$, 1--2 bit schemes are vacuous by many orders of magnitude; all mass $\ge 0$ is soundness again. \emph{Right}: ROC of the raw bound as a danger classifier (danger $=$ TV$>$0.2)---8-bit schemes detect essentially all danger at near-zero false-alarm rate, while sign-1bit tracks the diagonal, i.e.\ with the current witness its per-step score carries no signal (tightening it is open problem~(1)).}\label{fig:f5}
\end{figure}

\subsection{Tier B: a sub-Gaussian certificate for a controlled quantizer}
\label{sec:tierb}
\textbf{Subtractive dither.} We quantize with $\hat x = s\,(\mathrm{round}(x/s + \xi) - \xi)$ where the dither $\xi \sim \mathcal{U}[-\tfrac12,\tfrac12)$ is regenerated deterministically at read time. That the residual is then independent of the input and uniform per channel is classical dither theory~\cite{lipshitz1992,gray1993}; our contribution is wiring it into a KV-cache kernel under a proof--kernel contract and propagating it to a TV certificate. We choose $s = \mathrm{amax}_{\mathrm{blk}}/(Q-\tfrac12)$, with $Q=127$ the top INT8 level (so $Q-\tfrac12 = 126.5$), so that clipping \emph{never} occurs on non-outlier channels---a prerequisite for the assumption below.

\begin{lemma}[Uniform sub-Gaussian proxy]\label{lem:subg}
For $X \sim \mathcal{U}[-a,a]$ we have $\mathbb{E}[e^{\lambda X}] = \sinh(\lambda a)/(\lambda a) \le e^{\lambda^2 a^2/6}$. Hence the residual of subtractive dither with step $s$ satisfies $\mathbb{E}[e^{\lambda X}] \le e^{\lambda^2 s^2/24}$ and is sub-Gaussian with proxy exactly $s^2/12$, i.e.\ equal to its variance.
\end{lemma}
By Lemma~\ref{lem:subg}, the logit error has variance proxy $\varsigma^2_t = \tfrac{1}{d}\sum_c q_c^2\, s_{c,t}^2/12$ (the scale $s_{c,t}$ is block-constant in $t$; the same quantity is computed in-kernel as $\varsigma^2_{h,t}$ in Sec.~\ref{sec:storage}) and, with probability at least $1-\delta_{\mathrm{loc}}$ jointly over $S$ tokens, $|\varepsilon_t| \le u_t := \sqrt{2 \varsigma^2_t \log(2S/\delta_{\mathrm{loc}})}$, which is then propagated by Theorem~\ref{thm:eform} as $A = \sum_t \tilde p_t e^{u_t}$. Crucially, $V_t$ is computable from the stored quantization scales and the current query alone: \textbf{no witness needs to be stored}---the price of Tier A's 32 B/tok/head is paid only by black-box schemes.

\textbf{Request-level budget.} We allocate $\delta_{\mathrm{loc}} = \delta_{\mathrm{req}}/(L\cdot H\cdot T)$ over layers, heads and decode steps, so the guarantee is stated at the request level. Because the dependence is $\sqrt{\log(1/\delta)}$, tightening $\delta$ by orders of magnitude is inexpensive---quantified in Sec.~\ref{sec:experiments}.

\textbf{Scope of the probabilistic guarantee: non-adaptive queries.} The sub-Gaussian argument conditions on the query, i.e.\ it treats $q$ as statistically independent of the stored dither. That holds exactly for \emph{non-adaptive} queries: prefill queries over a given text, teacher-forced or otherwise fixed token trajectories, and any evaluation in which the trajectory is not itself a function of the compressed cache. In \emph{free-running} decoding it does not hold: the query at step $t$ is a continuous function of earlier attention outputs and hence of the very residuals being bounded, an adaptive dependence that no union bound repairs. We state this as an explicit assumption rather than hide it, with three consequences. (i) \emph{Tier A is adaptive-safe}: the witness bound of Theorem~\ref{thm:tiera} is a worst-case inequality valid for \emph{every} query, adaptive or not, so the universal meter---and every gating result in Sec.~\ref{sec:gatedeval}, all of which are driven by Tier-A witnesses---carries no adaptivity caveat. (ii) Beyond the fixed-trajectory zero-violation evidence (200-seed request-level counting, the kernel soundness checks), we have now \emph{measured} adaptive violations directly under free-running decoding ($8\times$H200, Qwen2.5-7B, 50 prompts $\times$ 20 dither seeds, prefill 4096, decode 64, $\delta_{\mathrm{req}}{=}10^{-2}$; at every step the true TV of every (layer, head) cell is computed by a parallel exact forward and compared against the runtime certificate): \textbf{0/1000} adaptive requests violated, 0/1000 in the teacher-forced control, over \textbf{100{,}352{,}000} monitored cells in total (numerical tolerance $10^{-9}$; the zero-tolerance residual consists entirely of margins $\le 1.8\times10^{-13}$, i.e.\ double-precision rounding, reported as such). Adaptivity was genuinely realized---\textbf{62.9\%} of free trajectories diverged from the fp16 baseline (median divergence step 22)---and the worst-case TV-to-bound ratio was \textbf{0.536}; rule-of-three puts the request-level violation rate below 0.3\% at 95\% confidence. We state the epistemic status plainly: this is validation, not proof. (iii) A martingale-style adaptive analysis, or a randomness-splitting design in which the certificate consumes dither that the generation path never touches, would restore the formal guarantee for free decoding; both remain open (and the same gap is unaddressed in prior runtime-certificate work).

\textbf{Outlier RoPE pairs.} A small number of frequency \emph{pairs} (selected per block by pair scale energy) are kept in FP16 and excluded from the variance term. This costs $16.1$ B/tok/head, i.e.\ $+6.3\%$ over the FP16 key ($256$ B) and $+12.5\%$ over the INT8 payload ($129$ B), and is the single largest lever on coverage: the ablation shows that without it, sub-Gaussian authorization collapses to nearly zero.

\textbf{Value side and the joint event.} Values carry no RoPE and their error propagates linearly, contributing the deterministic term $\sum_t \tilde p_t \rho^v_t$ of the output bound above. The joint event is $\mathrm{pass}_K \wedge \mathrm{pass}_V$, evaluated per (layer, KV head, position); we report joint coverage rather than inferring it from marginals. On the event $E_K$ the joint output bound is $\|y-\hat y\| \le 2\,\mathrm{TV}_K \max_t\|v_t\| + \sum_t \tilde p_t \rho^v_t$.

\subsection{Saturation, clamping, and a two-sided tightening}
\label{sec:twosided}
Since $\mathrm{TV}\le 1$, the sound meter value is $\min\!\bigl(1,\tfrac12(A^2-1)\bigr)$; raw values above saturation carry no guarantee, but remain an empirically discriminative risk score (the fallback experiment of Sec.~\ref{sec:closedloop} shows a $27.6\times$ error reduction precisely on the cases the raw score selects). This distinction is load bearing for the gating experiments of Sec.~\ref{sec:gatedeval} and is marked there explicitly: \emph{certified} gating operates at $\tau<1$, \emph{risk-ranked} gating at $\tau\ge1$.

The one-sided $e$-form is quadratically punished by a single large $c_t$. A second application of the same reasoning gives a companion bound that is linear in that regime:

\begin{proposition}[Two-sided form]\label{prop:twosided}
Let $A = \sum_t \tilde p_t e^{c_t}$ and $A^- = \sum_t \tilde p_t e^{-c_t}$. If $|\varepsilon_t| \le c_t$ for all $t$, then
\[
\mathrm{TV}(p,\tilde p) \;\le\; \tfrac12\Bigl(\frac{A}{A^-} - \frac{A^-}{A}\Bigr),
\]
and the meter $\min\!\bigl(1,\ \tfrac12(A^2-1),\ \tfrac12(A/A^- - A^-/A)\bigr)$ is sound.
\end{proposition}
\begin{proof}
With $Z = \sum_t \tilde p_t e^{-\varepsilon_t}$ we have $p_t/\tilde p_t = e^{-\varepsilon_t}/Z$ and $A^- \le Z \le A$, so $p_t/\tilde p_t \in [\,e^{-c_t}/A,\ e^{c_t}/A^-\,]$, an interval containing $1$ (because $A^- \le 1 \le A$). Hence $|p_t/\tilde p_t - 1| \le (1 - e^{-c_t}/A) + (e^{c_t}/A^- - 1)$, and taking the $\tilde p$-expectation gives $\mathrm{TV} \le \tfrac12\bigl((1 - A^-/A) + (A/A^- - 1)\bigr) = \tfrac12(A/A^- - A^-/A)$. Neither form dominates the other (at $A=2$, $A^-=0.6$ the two-sided form is looser), so the minimum is taken.
\end{proof}

The point is the massive-activation regime: if one token carries $\tilde p$-mass $\eta$ with a huge bound $C$, write $M = \eta\, e^{C} \gg 1$; then $A^2 - 1 \approx M^2 + 2M$ grows \emph{quadratically} in $M$ while $A/A^- - A^-/A \approx M$ grows \emph{linearly}. This is exactly the mechanism that drives the meter to $1.6\times10^{2}$ in Sec.~\ref{sec:e2e} and the regime where the capped det-tanh shape currently wins the traffic comparison (Sec.~\ref{sec:boundshape}). Structurally, $A^-$ is a softmax-weighted average exactly like $A$, so it merges across split-KV partitions under the same LSE weights (the argument of Sec.~\ref{sec:sglang}) and would cost one extra accumulator.

\textbf{Measured: the tightening does not materialize on real activations, and we report this negative result as such.} We evaluated Proposition~\ref{prop:twosided} offline on the six-quantizer harness of Fig.~\ref{fig:f5} ($n{=}4480$ cells per scheme; \texttt{experiments/out/p37\_twosided\_bound.json}), including the per-token retention form $\tfrac12\sum_t \tilde p_t \min\bigl(Ae^{c_t}-1,\ \max(e^{c_t}/A^- - 1,\ 1-e^{-c_t}/A)\bigr)$, which dominates both aggregates. In the informative regime (rtn-8bit) the aggregate two-sided form is $1.7\times$ \emph{looser} at the median and wins on only 1.6--1.7\% of cells; the per-token min coincides with the one-sided bound at the median (shrink $1.00\times$) and lifts coverage at $\tau{=}0.2$ only from 0.445 to 0.446. In the saturated regime the raw score shrinks by up to $26{,}433\times$ (rtn-2bit median; log-domain evaluation, no truncation of $c$) yet remains astronomically above $1$ (median $10^{12.5}$, maximum $10^{266}$), so the reported meter is pinned at saturation either way. The measurement refutes the single-spike picture on real data: massive-activation \emph{channels} make many $c_t$ large simultaneously, so $A^- = \sum_t \tilde p_t e^{-c_t}$ collapses toward $0$ and $A/A^-$ explodes. Numerical soundness of the min-form held on all 26{,}880 cells (0 violations). Consequently we do \emph{not} add the extra accumulator to the kernel; Proposition~\ref{prop:twosided} remains stated because it is a theorem, and applies to workloads whose residual bounds are genuinely single-spiked.

\textbf{Two negative results on bound selection.} (i) Applying a sub-Gaussian tail bound \emph{per token} and taking a union bound over the maximum is counterproductive: the resulting $\kappa = \sqrt{2\log(2SN/\delta)} \approx 6$ (with $N$ the number of per-token events) exactly cancels the $\sqrt{P_j}$ concentration gain (with $P_j$ the attention mass of block $j$). (ii) A Bernstein bound is dominated by its range term (linear in the residual range), which is inflated by massive-value channels; the uniform sub-Gaussian form has no such term and is strictly tighter here. Both mistakes were ours, and both were caught only by measurement.

\section{Experiments}
\label{sec:experiments}
Unless stated otherwise: Qwen2.5-7B~\cite{qwen25}, Mistral-7B~\cite{mistral7b} and Yi-1.5-6B~\cite{yi2024}, three domains (natural text, code, synthetic retrieval), ${\sim}50$ documents, 8k context, online block-wise scaling (no whole-sequence statistics), 8-bit K$+$V, $\tau_K = 0.2$, $\tau_V = 0.05$, $\delta_{\mathrm{req}} = 10^{-2}$, statistics at KV-head \emph{any-of-$G$} granularity (an unadorned $\delta$ always means $\delta_{\mathrm{req}}$).

\textbf{Risk--coverage curve (main result).} Table~\ref{tab:riskcov} reports joint K$+$V step coverage, single run per cell, with the deterministic $\tanh$ bound placed at $\delta = 0$.

\begin{table}[!ht]\centering\small
\caption{Risk--coverage curve: joint K$+$V step coverage vs.\ request-level budget $\delta$ (single run per cell; the deterministic tanh bound is the $\delta{=}0$ reference).}\label{tab:riskcov}
\begin{tabular}{llcccc}
\toprule
Model & Domain & $\delta{=}10^{-4}$ & $\delta{=}10^{-2}$ & $\delta{=}5{\times}10^{-2}$ & $\tanh$ ($\delta{=}0$) \\
\midrule
Qwen2.5-7B & natural & 0.792 & 0.809 & 0.814 & 0.167 \\
Qwen2.5-7B & code    & 0.772 & 0.796 & 0.803 & 0.118 \\
Qwen2.5-7B & needle  & 0.802 & 0.810 & 0.815 & 0.485 \\
Mistral-7B & natural & 0.792 & 0.793 & 0.793 & 0.520 \\
Mistral-7B & code    & 0.756 & 0.760 & 0.760 & 0.372 \\
Mistral-7B & needle  & 0.649 & 0.653 & 0.654 & 0.513 \\
Yi-1.5-6B$^{a}$ & natural & 0.557 & 0.559 & 0.559 & 0.141 \\
Yi-1.5-6B$^{a}$ & code    & 0.539 & 0.544 & 0.545 & 0.121 \\
Yi-1.5-6B$^{a}$ & needle  & 0.645 & 0.648 & 0.650 & 0.417 \\
\bottomrule
\end{tabular}

\smallskip
{\small $^{a}$\,Yi-1.5-6B is measured at 4k context because its \texttt{max\_position\_embeddings} is 4096; the other two families are at 8k.}
\end{table}

\textbf{The Yi row is our only weakening cross-model result, and we report it in full.} Yi's coverage (54.4--64.8\% at $\delta{=}10^{-2}$) is markedly lower than Qwen's (79.6--81.0\%) and Mistral's (65.3--79.3\%). This is \emph{not} explained by the shorter context: the certificate's $\sqrt{\log S}$ dependence means 4k should be \emph{more} favourable than 8k, so the effect is real rather than an artefact of the measurement. The relative gain still holds---page-in rate falls 39.7--48.7\% against tanh, inside the 28.7--77.0\% band of the other two families---and $\delta$ sensitivity is the \emph{lowest} of the three (0.22--0.62\,pp across a $500\times$ range). The third family also yields \textbf{zero violations}, a third independent soundness check. The headline coverage range over the nine model$\times$domain cells is therefore \textbf{54.4--81.0\%} at the single working point $\delta{=}10^{-2}$ (the earlier two-family range 64.9--81.5\% is superseded; we do not quote a cross-$\delta$ envelope). The mechanism is \emph{not} a property of the model but an interaction of block-level scaling with Yi's diffuse channel energy; it is identified experimentally in Sec.~\ref{sec:yianomaly} and disappears under the per-token scaling used in deployment, where Yi is in fact \emph{easier} than Qwen.

Reading: (a) the relative reduction in page-in rate is \textbf{28.7--77.0\%} and holds across all three model families and all three domains; (b) tightening $\delta$ by $500\times$ costs at most 3.2\,pp of coverage (median 0.47\,pp over the nine model~$\times$~domain cells)---the slope of the curve answers directly the objection that the gain is merely bought by relaxing the failure probability; (c) heterogeneity is reported honestly: the deterministic bound fares better on Mistral, and four of the nine cells exceed our 30\% page-in target, all on the two weaker families (Mistral needle 34.7\%; Yi 35.2--45.6\% across its three domains).

\textbf{Certificate-formula head-to-head} (global scaling, identical quantizer and outlier handling): the deterministic bound authorizes 15.14 raw bits with a 44.4\% KV-head page-in rate; ours authorizes 11.71 raw bits (11.98 physical, including outliers) at 22.0\%. Stated precisely: \emph{under an explicit $1\%$ request-level risk budget, the probabilistic certificate halves the page-in rate and authorizes compression 3.43 bits/dim deeper ($15.14 \to 11.71$ raw bits).} The two guarantees are of different strength and we do not claim equivalence.

\textbf{Joint K$+$V} (Qwen2.5-1.5B, frozen early setting). Marginals $\mathrm{pass}_K = 89.1\%$, $\mathrm{pass}_V = 97.1\%$; the measured joint is $86.7\%$, inside the interval $[86.2\%, 89.1\%]$ implied by the marginals. KV-head joint page-in is 27.6\%; the idealized amortized effective width is 10.39 bits, charging page-in steps at full FP16 price.

\textbf{Violations (scope limited).} Across the first two model families (the third family adds its own zero-violation run, reported above), three domains, all $\delta$ levels, both certificates, and 200 independent dither seeds counted at request granularity (single text, 64 scattered query positions, all layers and all heads, 8-bit setting), we observed no violations; the rule-of-three 95\% upper bound is 1.5\%, consistent with $\delta_{\mathrm{req}} = 1\%$. Soundness follows from the theorems and their assumptions; zero observed violations is validation, not proof.

\textbf{Real memory saving (measured, not computed).} Our earlier byte account was \emph{computed}. We now build the real packed store (int8 payload $+$ FP16 scale $+$ int16 index $+$ FP16 outliers, with the dither reconstructed from Philox and therefore stored at zero cost) and measure it with \texttt{torch.cuda.memory\_allocated}: at $S = 32768$, $d = 128$, one KV head, the packed store occupies $4{,}753{,}408$ B versus $8{,}388{,}608$ B for FP16, a \textbf{real $43.3\%$ saving}. As a control, the ``dequantize back to bf16'' simulated setting measures a saving of \textbf{0.0\%}, confirming the red-team's objection. The full chain (packed store $\to$ attention kernel $\to$ certificate) runs end to end: at $\tau = 0.2$ the coverage is $0.792$ and the relative output error $9.0\!\times\!10^{-3}$. The saving is independent of context length ($43.3\%$ at 8k, 32k and 128k alike).

\textbf{Ablations} appear where they are used: the outlier-pair count $m$ (Secs.~\ref{sec:designspace}, \ref{sec:bypass}), scale granularity (Sec.~\ref{sec:storage}), and bound shape (Secs.~\ref{sec:twosided}, \ref{sec:boundshape}); the deterministic witness-storing variant \texttt{cert\_D} is subsumed by Tier A.

\textbf{Limitations.} The certificate is a per-layer, per-head \emph{local} guarantee (the value-side bound is deterministic given $\tilde p$, so the joint $\mathrm{pass}_K \wedge \mathrm{pass}_V$ event consumes no additional probabilistic budget beyond the K side); the Tier-B probabilistic guarantee is stated for non-adaptive queries only (Sec.~\ref{sec:tierb})---free-running decoding is outside the theorems, and Tier A is the adaptive-safe tier; certificate aggregation under multi-GPU / tensor parallelism is unverified; model coverage is three 6--7B-class GQA families for the risk--coverage evidence (systems experiments are Qwen-only). \emph{Limitations now lifted}: ``no real packed cache'' is lifted by the measurement above; ``PCIe traffic is a simulated bill'' is lifted by Sec.~\ref{sec:closedloop} (real page-in loop) and Sec.~\ref{sec:blockpage} (block-level paging, $128\times$ less traffic); ``paired significance testing pending'' is lifted by Sec.~\ref{sec:ruler}.

\subsection{Long chain-of-thought workloads (model self-generated, not human-written text)}
\label{sec:cot}
A reasonable objection is that coverage was measured on human-written documents, whereas the long CoT of a reasoning model has a different attention structure (self-generated, repetitive patterns, possibly sharper distributions), so the certificate might collapse there. We generated 4 real long CoT traces with DeepSeek-R1-Distill-Qwen-7B~\cite{deepseekr12025} (1724--4096 tokens, 12{,}432 tokens in total) and measured:

\begin{table}[!ht]\centering\small
\caption{Coverage on model-generated long CoT (DeepSeek-R1-Distill-Qwen-7B; $n{=}4$, a no-collapse check, not a significance result).}\label{tab:cot}
\begin{tabular}{ccccc}
\toprule
Sample & Generated length & Coverage (query head) & Coverage (KV head) & Attention entropy \\
\midrule
0 & 4096 & 0.857 & 0.800 & 3.79 \\
1 & 2617 & 0.857 & 0.800 & 3.66 \\
2 & 1724 & 0.857 & 0.800 & 3.50 \\
3 & 3995 & 0.858 & 0.798 & 3.66 \\
\bottomrule
\end{tabular}
\end{table}

Mean KV-head coverage is \textbf{0.800}, on par with the 0.792 full-chain coverage of Sec.~\ref{sec:experiments} on ordinary long text, and \textbf{no collapse was observed}; the variance across the four samples is very small. Note $n = 4$: this is a ``no collapse observed'' statement, not a significance result.

\subsection{Long-context scaling to 128k}
\label{sec:s128k}
The introduction motivates with long context, so we measure how coverage itself scales with $S$, on the one model that fits full-activation capture at 128k on our GPU (Qwen2.5-1.5B-Instruct; two concatenated domain streams, natural and needle; the Sec.~\ref{sec:experiments} configuration at $\delta_{\mathrm{req}}{=}10^{-2}$; the model's own rotary $\cos/\sin$ are hooked, so YaRN and non-YaRN runs are exactly faithful to the forward pass). Joint K$+$V step coverage: \textbf{63.7--64.5\%} at 8k, \textbf{57.3--60.0\%} at 32k, and \textbf{26.6--41.3\%} at 128k, the last under the vendor's official YaRN~\cite{yarn2024} factor-4 long-context configuration since the native window is 32k. Two readings, both honest. (i) The reviewer-friendly prediction ``coverage will hold because the radius grows as $\sqrt{\log S}$'' is \emph{not} confirmed: the radius grows only ${\sim}5\%$ from 32k to 128k, yet coverage drops by up to half---the decline is a property of long-context attention statistics under YaRN, not of the certificate mathematics, and we report it as such. (ii) The \emph{relative} picture strengthens: the deterministic tanh reference collapses to at most \textbf{0.2\%} coverage at 128k, while the probabilistic certificate still authorizes 26.6--41.3\% of steps. Zero violations at all three lengths (a fourth independent soundness check).

\textbf{On a native long-context model the decline does not reproduce (same protocol, $8\\times$H200).} Re-running with Llama-3.1-8B-Instruct, whose native window is 131{,}072 (no YaRN extrapolation of any kind), joint coverage on the needle domain is \textbf{77.0 / 77.1 / 77.0\%} at 8k/32k/128k and 71.7/73.8/71.4\% on natural---\emph{flat in sequence length}, with zero violations at every length, while the tanh reference still collapses (natural $0.399 \to 0.137$). Read together, the two experiments localize the cause: the 128k decline above is a property of \emph{YaRN-extrapolated} attention statistics (small model, extended window), not of long context per se; native long-context models show no such decline.

\section{System Implementation and On-Device Results}
\label{sec:system}
This section is organized in four parts: architecture (\ref{sec:arch}), implementation (\ref{sec:impl}), evaluation (\ref{sec:eval}), and analysis (\ref{sec:analysis}). The loop being implemented is:

\medskip
\noindent\fbox{\parbox{0.965\linewidth}{\small
\textbf{The meter loop} (one decode step, one layer, one KV head).\\[2pt]
\emph{Write path} (once, at cache write): quantize $k \to \hat k$ with any scheme; Tier A stores the witness $w_{t,b} = \|r_{t,b}\|$ (32 B/tok/head); Tier B (dithered INT8) stores payload and scales only.\\[2pt]
\emph{Read path} (fused into decode attention, shared data load): accumulate $A = \sum_t \tilde p_t e^{c_t}$ with $c_t$ from Theorem~\ref{thm:tiera} (Tier A) or the sub-Gaussian radius $u_t$ (Tier B); merge across KV splits under the LSE weights; meter $= \min\!\bigl(1, \tfrac12(A^2-1)\bigr)$.\\[2pt]
\emph{Gate} (per (layer, head) or per request): if the meter (or, above saturation, the raw risk score) exceeds $\tau$, page in the top-contributing blocks (\ref{sec:blockpage}) or the whole head from the exact backing store, recompute, write back, zero the witness; otherwise serve the compressed path.}}
\medskip

\begin{figure}[t]\centering
\includegraphics[width=0.98\linewidth]{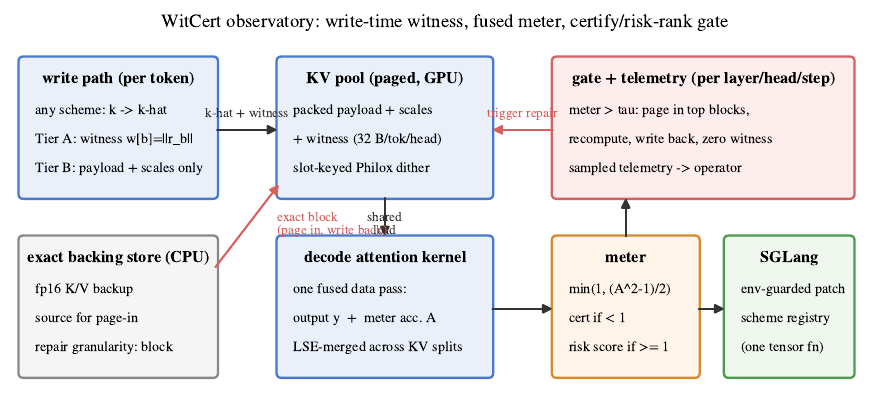}
\caption{The WitCert observatory. Any scheme's residual leaves a witness at write time; the decode kernel computes output and meter in one fused data pass (LSE-merged across KV splits); the gate repairs at block granularity from the exact backing store and zeroes the repaired witnesses; SGLang integration is a bit-identical-when-off env-guarded patch.}\label{fig:f7}
\end{figure}

\subsection{Architecture}
\label{sec:arch}
The meter loop above is the design; Fig.~\ref{fig:f7} shows where it lives in a serving engine. The rest of this section is organized as implementation (\ref{sec:impl}), evaluation (\ref{sec:eval}) and analysis (\ref{sec:analysis}).

\subsubsection{Why the meter is structurally cheap by design}
\label{sec:cheap}
Scales are block wise ($\mathrm{BLK} = 256$) and the tile width divides the block width, so within a tile $\varsigma^2$ is \emph{token independent}; the certificate accumulator therefore collapses from a $[\mathrm{BLOCK\_M}, \mathrm{BLOCK\_N}]$ matmul to one scalar per query times the softmax row sum. The dither PRNG is likewise shared by the $\mathrm{BLOCK\_M}$ queries in a tile, amortizing its cost by a factor $\mathrm{BLOCK\_M}$. (The read side uses the same five-tuple Philox as the write side; an early version used \texttt{tl.rand}, was therefore not the same random stream as the write side, and was flagged as a violation by the contract test. It has been fixed.)

\textbf{The marginal cost of the certificate itself, measured in two places.} (i) In our research-grade kernel it already falls into the noise, $-1.3\%$ (the latency difference between certificate on and off; see \texttt{docs/active\_block\_policy.md}). (ii) Inside the SGLang 0.5.9 production decode kernel it is $+11.9\%$ (Sec.~\ref{sec:sglang}). The two are not in conflict: our kernel's latency is dominated by dither reconstruction ($77.9\%$ of it), so the certificate is a small quantity in its shadow, whereas the SGLang kernel has no dither reconstruction and is tightly tuned for occupancy, so one extra accumulator shows up directly as $+11.9\%$. \textbf{We adopt the latter as the external figure}---it is the number a production system will see.

\subsection{Implementation}
\label{sec:impl}
\subsubsection{Research kernels and the proof--kernel contract}
\label{sec:kernels}
The write kernel implements $\hat x = s\,(\mathrm{round}(x/s+\xi)-\xi)$ with the dither $\xi$ drawn by Philox4x32-10~\cite{philox2011} from a five-tuple counter (request, layer, KV head, token, channel), so the random stream is invariant to batch scheduling and CUDA-graph replay and can be reconstructed exactly at read time. The scale $s = \mathrm{amax}_{\mathrm{blk}}/126.5$ (i.e.\ $Q-\tfrac12$) guarantees \emph{strictly no clipping} on non-outlier channels, which is the premise of the sub-Gaussian assumption. Two implementation facts turned out to be load bearing. First, quantization must use the \emph{stored} FP16 scale: quantizing with an fp32 scale while storing FP16 breaks $|\varepsilon| \le s/2$ (we measured 1639 violations) and with it the sub-Gaussian assumption---a concrete case of a kernel detail deciding whether a theorem holds. Second, Triton's \texttt{uint32} multiply/broadcast semantics diverged from our reference at round two of Philox; we moved both sides to int64 with explicit masking, trading throughput for semantic certainty. The read kernel performs INT8 dequantization, Philox dither reconstruction and certificate accumulation in a single kernel, sharing one data load with attention.

\textbf{Acceptance (nine items, 8/9 complete).} Elementwise agreement between the CPU reference and the GPU kernel (int8 disagreement $1/262{,}144$, asserted to fall on a 0.5 rounding tie, where both rounding directions satisfy $|\mathrm{err}| \le s/2$; attention output $\max|\Delta o| = 2.9\!\times\!10^{-6}$, certificate $\max|\Delta| = 8.9\!\times\!10^{-8}$); five-tuple stream uniqueness; CUDA-graph replay \emph{bit identical} to eager execution with five interleaved replays agreeing exactly (the seed is unaffected by scheduling), and a certificate produced at \emph{every} graph-captured decode step---this closes the gap left by a Python-level probe, which CUDA graph rejects precisely because of its host synchronization; zero clipping on non-outlier channels; measured kernel error fed directly to the certificate with zero violations; the physical byte account exact to the byte; overhead reported as three separate items; and the quad report. The tiled kernel is covered by a separate contract test ($\Delta o = 9.8\!\times\!10^{-6}$, $\Delta\mathrm{cert} = 6.0\!\times\!10^{-7}$, zero violations). Not complete: production-grade kernel optimization (warp specialization, asynchronous prefetch).

\subsubsection{The SGLang decode kernel}
\label{sec:sglang}
The kernels above are all our own research-grade implementations. To answer ``can the certificate go into a real inference engine'', we added the certificate accumulator to the Triton decode attention of \textbf{SGLang 0.5.9}~\cite{sglang2024} (the grouped two-stage kernel of \texttt{decode\_attention.py}), a \textbf{500-line} diff, in \texttt{integration/witcert\_sglang.patch}.

\textbf{Key structure: the certificate shares the split-KV reduction with the output.} SGLang uses a flash-decoding-style~\cite{flashdecoding2023} two-stage structure (descended from FlashAttention~\cite{flashattention2022}): stage 1 emits per-split local $(\mathrm{acc}/e_{\mathrm{sum}}, \mathrm{lse})$, and stage 2 merges them with LSE weights. But the certificate $A = \sum_t \tilde p_t e^{u_t}$ \emph{is itself a softmax-weighted average}, so the per-split $A_s = \sum_{t \in s} \tilde p_t e^{u_t}/e_{\mathrm{sum},s}$, merged under the same LSE weights, yields \emph{exactly} the global $A$---no extra synchronization, no second scan. This is not a coincidence: the certificate and the attention output are both linear functionals of $\tilde p$.

To be compatible with token-granularity paging, this integration uses \textbf{one scalar scale per token} (rather than the per-block, per-channel scale of the main text). The certificate formula is unchanged, and since $u_t = \alpha_h\, s_t$ is \emph{linear} in the token scale $s_t$ (with $\alpha_h$ a per-(query, head) constant), the square root can be hoisted out of the loop entirely.

\begin{table}[!ht]\centering\small\setlength{\tabcolsep}{4pt}
\caption{Acceptance of the SGLang decode-kernel integration.}\label{tab:sglaccept}
\begin{tabular}{@{}p{0.56\linewidth}p{0.38\linewidth}@{}}
\toprule
Acceptance item & Result \\
\midrule
Certificate off vs upstream kernel & \textbf{bit identical} ($\max|\Delta| = 0$) $\Rightarrow$ zero intrusion \\
Certificate vs torch reference (shuffled paged slots, 8-way split) & max relative error $\mathbf{5.06\!\times\!10^{-6}}$ \\
Overhead & $\mathbf{+11.9\%}$ \\
\bottomrule
\end{tabular}
\end{table}

\textbf{Attributing the overhead (discriminative experiment).} Keeping the accumulator but removing the scale load and the exponential still costs $+11.77\%$; inlining the scale into the K row (saving one memory transaction) reduces it by only 0.6\%; and aligning quantization blocks to pages so that $u$ is constant within a tile (reducing the number of exponentials from $\mathrm{BLOCK\_H}\times\mathrm{BLOCK\_N}$ to $\mathrm{BLOCK\_H}$) yields no gain. \textbf{So this 11.9\% is almost entirely the register/occupancy cost of ``one more accumulator'', and has nothing to do with the certificate's arithmetic.} This does not contradict the noise-level ($-1.3\%$) certificate cost of our own kernel in Sec.~\ref{sec:cheap}: that kernel's latency is dominated by dither reconstruction (77.9\% of it), whereas the SGLang kernel has no dither reconstruction and is tightly tuned for occupancy, so any new accumulator shows up immediately. \textbf{The external figure should be $+11.9\%$}---that is the number a production system will see.

At this stage only certificate \emph{computation} is integrated; storage is the subject of the next subsection.

\subsubsection{Storage side: a paging-compatible certified KV cache}
\label{sec:storage}

Computing the certificate is not enough: its soundness requires \emph{subtractive dithering}, and SGLang's existing fp8 KV cache uses deterministic rounding, which violates the theorem's hypothesis. This is exactly what the proof--kernel contract is for---one cannot bolt a certificate onto a kernel that does not produce the assumed error distribution. A certified deployment must therefore bring its own storage.

\textbf{The scale granularity has to be re-chosen.} The body of this paper uses per-256-token-block, per-channel scales. SGLang's pool is paged at \emph{token-slot} granularity (in the PagedAttention style~\cite{vllm2023}): slots are assigned arbitrarily by the allocator, consecutive tokens of a request need not be contiguous, and during decoding a page fills \emph{incrementally}, so when token $t$ is written the later tokens do not exist and no block amax can be computed. We therefore scale \emph{per token, per group of $G_c$ channels}---a granularity that depends only on the current token and is fully decoupled from slot assignment. On Qwen2.5-7B real activations (8k, 8-bit K$+$V, no outlier bypass):

\begin{table}[!ht]\centering\footnotesize\setlength{\tabcolsep}{4pt}
\caption{Scale granularity on real activations (Qwen2.5-7B, 8k, 8-bit K$+$V, no bypass; real-V output error; single run).}\label{tab:granularity}
\begin{tabular}{lccccc}
\toprule
Scheme & Rel.\ output error & vs.\ ours & K$+$V B/tok/head & Saved & Paging-safe \\
\midrule
block $\times$ channel (paper body) & $2.449\!\times\!10^{-2}$ & $1.00\times$ & 258 & 49.6\% & no \\
token $\times$ 16-channel groups & $\mathbf{1.821\!\times\!10^{-2}}$ & $\mathbf{0.74\times}$ & 288 & 43.8\% & yes \\
\textbf{token $\times$ 32-channel groups} & $2.131\!\times\!10^{-2}$ & $\mathbf{0.87\times}$ & 272 & 46.9\% & yes \\
token scalar & $3.720\!\times\!10^{-2}$ & $1.52\times$ & 260 & 49.2\% & yes \\
\bottomrule
\end{tabular}
\end{table}

\textbf{This is not a compromise---it is more accurate.} Per-token scaling tracks the very-large-magnitude tokens (the massive-activation effect discussed above), which a 256-token block scale smooths away. Paging-safety and quality point the same direction here. We take $G_c{=}32$ as the default.

\textbf{The dither counter uses the slot, not the sequence position.} The Philox five-tuple is (layer, kv\_head, \textbf{slot}, channel). Slots are stable in the cache---paging, prefix sharing and retraction never change the contents of a written slot, so the read side can always reconstruct the dither $\xi$ used at write time; sequence positions are not stable (under prefix sharing the same token appears at different positions in different requests).

Validation (RTX 4090, $S{=}4096$, batch 2, 28 query / 4 KV heads, 8-way split): Triton dither vs.\ the torch reference is \textbf{bit-identical}; residuals satisfy $|\mathrm{err}| \le s/2$ with \textbf{zero violations} (max ratio 1.0000); in-kernel dequantisation matches the reference to $6.43\!\times\!10^{-4}$ (the fp16 baseline itself is $2.17\!\times\!10^{-4}$); the in-kernel certificate matches to $4.59\!\times\!10^{-4}$; \textbf{soundness holds with 0/56 violations} (true TV $10^{-3}$--$10^{-2}$ against certificates 0.119--0.170, median tightness $22.2\times$); coverage at $\tau{=}0.2$ is 1.000; end-to-end quality cost $1.62\!\times\!10^{-2}$; latency fp16 $\to$ quantised $5.47\times$ $\to$ quantised${+}$certificate $5.49\times$; measured memory 136 B/tok/head \emph{per side} (K$+$V total 272 B, matching the $G_c{=}32$ row of Table~\ref{tab:granularity}), \textbf{46.9\% saved}.

\textbf{Structural finding: the certificate's marginal cost collapses once storage is quantised.} The same certificate costs \textbf{+11.9\%} on the fp16 storage path but only \textbf{+0.35\%} on the quantised path, because $\varsigma^2_{h,t} = (\mathrm{sm\_scale}^2/12) \sum_c q_{h,c}^2\, s_{c,t}^2$ is exactly one $[\mathrm{BLOCK\_H}, d] \times [d, \mathrm{BLOCK\_N}]$ \texttt{tl.dot} whose $s$ \emph{was already loaded for dequantisation}---the certificate rides along with zero extra memory traffic. This resolves the apparent tension with the ``nearly free'' claim earlier: the +11.9\% was measured on a kernel that needs no scales at all. \textbf{The certificate's cost should not be priced separately; it and the quantisation are two uses of the same load.}

\subsubsection{End-to-end in the serving process}
\label{sec:serving-e2e}

\textbf{Source patching is required; monkeypatching does not work.} SGLang's scheduler runs in a spawned subprocess, so parent-process patches never reach it (the subprocess simply exits with EOFError). Our patch is \textbf{327 lines} across three sites (\texttt{integration/\allowbreak witcert\_\allowbreak sglang\_\allowbreak serving.patch}), all guarded by an environment variable so that the disabled path is bit-identical to upstream: the pool constructor, \texttt{get\_\allowbreak cell\_\allowbreak size\_\allowbreak per\_\allowbreak token},
and \texttt{forward\_\allowbreak decode}.

\textbf{The middle site is easy to miss, and missing it zeroes the benefit.} Our first version changed only the pool and the backend; KV capacity did not move, because SGLang still sized it from \texttt{kv\_cache\_dtype} and the saved memory sat idle. \emph{Byte savings become serving capacity only if the scheduler is told about them.}

\subsubsection{Fixed-channel outlier bypass: two steps that must both be taken}
\label{sec:bypass}

The body of the paper selects $m$ outlier RoPE pairs per block. Under token-granular paging there are no blocks, but measurement shows a \emph{per-(layer, kv-head) fixed channel set} suffices: the massive-activation RoPE pairs are highly stable across tokens---the top-8 pairs chosen from the first half of a sequence match those from the second half with Jaccard \textbf{0.940} (Qwen2.5-7B) and \textbf{0.880} (Yi-1.5-6B), minimum 0.750 / 0.500. A fixed set removes the need for per-token indices: only the $2m$ fp16 values are stored (16 B at $m{=}4$).

\textbf{Bypassing alone is not enough---the bypassed channels must also be excluded from the group amax.} Otherwise the other 31 channels of the group still pay for the very large channel that is already bypassed, the scale does not drop, and the certificate barely improves. A controlled experiment (4 injected massive RoPE pairs):

\begin{table}[!ht]\centering\small\setlength{\tabcolsep}{4pt}
\caption{Bypass alone is not enough: the bypassed channels must also leave the group amax (controlled experiment, 4 injected massive RoPE pairs).}\label{tab:amax}
\begin{tabular}{lccc}
\toprule
Treatment & Median meter & Max meter & Coverage ($\tau{=}0.2$) \\
\midrule
No bypass & 1.0311 & 2.378 & 0.000 \\
Bypass, amax unchanged & 0.9608 & 2.177 ($-8\%$) & 0.000 \\
\textbf{Bypass + excluded from amax} & \textbf{0.0620} & \textbf{0.0731 ($32\times$ lower)} & \textbf{1.000} \\
\bottomrule
\end{tabular}
\end{table}

Coverage-versus-budget on real activations ($S{=}4096$, $\tau{=}0.2$, amax exclusion applied):

\begin{table}[!ht]\centering\small\setlength{\tabcolsep}{4pt}
\caption{Coverage vs.\ outlier budget on real activations ($S{=}4096$, $\tau{=}0.2$, amax exclusion applied; single run).}\label{tab:bypassbudget}
\begin{tabular}{lcccc}
\toprule
$m$ & Qwen2.5-7B & Yi-1.5-6B & B/tok/head & Saved \\
\midrule
0 & 0.525 & 0.907 & 136 & 46.9\% \\
2 & 0.818 & 0.999 & 144 & 43.8\% \\
\textbf{4} & 0.843 & \textbf{1.000} & 152 & 40.6\% \\
8 & 0.929 & 1.000 & 168 & 34.4\% \\
16 & \textbf{0.980} & 1.000 & 200 & 21.9\% \\
\bottomrule
\end{tabular}
\end{table}

\textbf{Prefix reuse (radix cache) is supported.} The extend/prefill kernel received the same in-kernel dequantisation treatment (\textbf{155-line} diff): bit-identical to upstream when disabled, and matching a ``dequantise then run upstream'' reference to $9.34\!\times\!10^{-6}$ when enabled. Across requests sharing a long prefix, outputs track the fp16 baseline closely (2 of 4 answers bit-identical, the rest semantically equivalent).

\textbf{The write path must also be a kernel.} Our first working version reached only 63.7 tok/s because quantisation on write was still in torch: ten rounds of int64 Philox tensor ops, inside a 28-layer per-step loop, generate a great many kernel launches. Moving it to Triton (\textbf{bit-identical} to the torch reference, $82\times$ faster per call) raised this to 575.9 (single development-milestone run; the final median-of-3 figure is Table~\ref{tab:serving}'s 511.8). \emph{Offline validation of the read path alone misses this bottleneck entirely.}

\textbf{Certificate telemetry in production.} Each decode step emits certificates, aggregated on the GPU and flushed by sampling. One full-sample run (1.5B, no bypass, request-level allocation): 249{,}312 (layer, query-head) certificates, \textbf{coverage 81.2\% at $\tau{=}0.2$}, with a few heads far above threshold (massive activation). An operator therefore learns directly that 18.8\% of (layer, head, step) triples carry no guarantee at the current threshold---the deployed form of the runtime observability layer proposed in the introduction.

The telemetry itself has three traps we hit in practice (single diagnostic measurements, used for localisation rather than quantification): a per-layer \texttt{.item()} sync dropped throughput from ${\sim}560$ to 409.9; a \texttt{.float().contiguous()} copy of the whole scale buffer inside the certificate branch brought it to 423.7; and even with both fixed, per-layer reductions still cost enough that we moved to sampling (stride 32). With those fixed, an isolated measurement puts the telemetry at 287.6 (off) vs.\ 284.1 (on) on 7B---\textbf{1.2\%}, meaning the 6.4\%/16.2\% certificate cost above is essentially all kernel, not telemetry. \emph{Observation code can easily cost more than what it observes and must be an order of magnitude cheaper than what it measures.}

When the certificate overflows fp32 its value is $+\infty$. This is \emph{still sound}---infinity is a valid upper bound, meaning ``no guarantee for this head''---and is counted as over-threshold; means and maxima are taken over finite values only and the overflow fraction is reported separately.

\textbf{Boundary (stated plainly).} Covered: both the decode and extend/prefill paths, \texttt{page\_size} 1 and 16, model scales 1.5B and 7B, radix cache on and off. \textbf{Not covered}: the FlashInfer~\cite{flashinfer2025} backend (this is a Triton-side contribution; FlashInfer is a separate CUDA kernel family needing its own implementation); multi-GPU / tensor parallelism is unverified.

\subsection{Evaluation}
\label{sec:eval}
\subsubsection{Quad report and same-framework comparison ($S = 32768$, $n_q = 24$, $\tau = 0.2$, $\delta_{\mathrm{req}} = 10^{-2}$; synthetic K)}
\label{sec:quad}
Table~\ref{tab:quad} compares four methods on four axes.

\begin{table}[!ht]\centering\footnotesize\setlength{\tabcolsep}{4pt}
\caption{Quad report, synthetic K, single run ($S{=}32768$, $n_q{=}24$, $\tau{=}0.2$). RTN-INT8$^{a}$ is the base form of KVQuant/KIVI.}\label{tab:quad}
\begin{tabular}{lccccc}
\toprule
Method & Cert.\ coverage & Fallback & Latency & B/tok/head & Rel.\ output error \\
\midrule
fp16 & --- & --- & 1.593 ms ($1.00\times$) & 256 & $2.7\!\times\!10^{-4}$ \\
RTN-INT8$^{a}$ & none & --- & 1.575 ms ($0.99\times$) & 129 ($-49.6\%$) & $8.5\!\times\!10^{-3}$ \\
RTN-INT4 & none & --- & 1.508 ms ($0.95\times$) & 65 ($-74.6\%$) & $1.6\!\times\!10^{-1}$ \\
\textbf{WitCert-INT8} & \textbf{79.2\%} & 75\% (any-of-6) & 8.627 ms ($5.41\times$) & 145.1 ($-43.3\%$) & $\mathbf{9.0\!\times\!10^{-3}}$ \\
\bottomrule
\end{tabular}
\end{table}

Three readings. (a) \emph{No measurable quality cost from the certificate}: WitCert and RTN-INT8 are at the same error level ($9.0\!\times\!10^{-3}$ vs $8.5\!\times\!10^{-3}$), and the extra memory is the $16.1$ B/tok/head FP16 outlier bypass ($+12.5\%$ over the $129$ B of RTN-INT8, $+6.3\%$ over the $256$ B FP16 key), bought in exchange for a runtime TV upper bound. (b) The $5.41\times$ latency comes from the contract-compliant dequantization path (two non-exclusive ablation deltas: Philox $77.9\%$, IEEE precision $30.4\%$; already improved $2.7\times$ by four-output reuse, with a further $1.48\times$ available from a TF32 knob), while \textbf{the marginal cost of the certificate itself is within the noise ($-1.3\%$)}; the near-zero overhead of the RTN path shows that memory traffic is already saturated and the dither PRNG is pure added compute, so the path to production is clear. (c) RTN-INT4 saves $74.6\%$ but incurs an error of $1.6\!\times\!10^{-1}$ (${\sim}19\times$ worse than RTN-INT8) while the system remains entirely unaware---which is exactly the risk visibility WitCert can supply to any compression scheme.

\subsubsection{Quad report on real activations, and 32k (replacing synthetic K)}
\label{sec:realact}
The quad report of Sec.~\ref{sec:quad} uses synthetic (random Gaussian) K; this section redoes the same measurement on \textbf{post-RoPE K activations of Qwen2.5-7B over real documents in three domains} (every fourth layer, all GQA query heads). $S = 8192$, 7 documents across three domains, $\tau = 0.2$, $\delta_{\mathrm{req}} = 10^{-2}$:

\begin{table}[!ht]\centering\small
\caption{Quad report on real Qwen2.5-7B activations ($S{=}8192$; error column is the K-side proxy; single run).}\label{tab:realact}
\begin{tabular}{lcccc}
\toprule
Method & Cert.\ coverage & Fallback (any-of-6) & B/tok/head & Rel.\ output error \\
\midrule
fp16 & --- & --- & 256.0 & --- \\
RTN-INT8 & none & --- & 129.0 ($-49.6\%$) & $4.27\!\times\!10^{-2}$ \\
RTN-INT4 & none & --- & 65.0 ($-74.6\%$) & $4.64\!\times\!10^{-1}$ \\
\textbf{WitCert-INT8} & \textbf{80.6\%} & \textbf{23.0\%} & 145.1 ($-43.3\%$) & $\mathbf{2.38\!\times\!10^{-2}}$ \\
\bottomrule
\end{tabular}
\end{table}

The fallback column is \emph{measured} any-of-6 (the harness groups six query heads per KV head), not inferred from per-head coverage: on synthetic K the per-head exceedances are nearly independent ($75\% \approx 1-0.792^6$, the quad report), whereas on real activations they are strongly correlated within a KV group ($23.0\% \ll 1-0.806^6 = 72.6\%$)---correlation that itself reflects the channel-concentrated structure of real keys.

\textbf{Dither is quality neutral ($2\times2$ ablation).} Measured with \textbf{real V} on real attention output (quantizer $\times$ outlier, Qwen2.5-7B, 8k, 7 documents across three domains, 8-bit K$+$V); an earlier ``$1.7$--$1.8\times$ better than RTN'' claim rested on a K-side proxy and is retracted (Appendix~\ref{app:evolution}):

\begin{table}[!ht]\centering\small
\caption{$2\times2$ ablation with real V: dither is quality neutral; the quality gain belongs to the outlier bypass.}\label{tab:ablation22}
\begin{tabular}{lcc}
\toprule
Configuration & Rel.\ output error & K$+$V B/tok/head \\
\midrule
RTN & $2.751\!\times\!10^{-2}$ & 258.0 \\
RTN $+$ outlier & $1.995\!\times\!10^{-2}$ & 290.1 \\
dither & $2.717\!\times\!10^{-2}$ & 258.0 \\
dither $+$ outlier & $1.951\!\times\!10^{-2}$ & 290.1 \\
\bottomrule
\end{tabular}
\end{table}

In this ablation \emph{both} K and V carry outlier pairs (hence 290.1 B/tok/head), unlike the RULER configuration where only K does. Decomposition: \textbf{the gain from dither itself is only $1.01\times$, i.e.\ quality neutral}; the entire $1.38\times$ gain comes from the outlier bypass, at a cost of $+12.4\%$ memory. So the correct conclusion is not ``guarantee and quality move together'', but rather: \textbf{subtractive dither is quality neutral relative to same-bit-width RTN, and it is a necessary condition for the probabilistic certificate $\Rightarrow$ the probabilistic certificate is obtained at zero accuracy cost}; the quality gain comes from the outlier bypass, is unrelated to the certificate, and is an independently selectable knob. This conclusion is weaker than the original, but defensible.

\textbf{The 32k long-context setting} (same configuration, 4 documents across three domains):

\begin{table}[!ht]\centering\small
\caption{The 32k long-context setting (same configuration as Table~\ref{tab:realact}; K-side proxy error).}\label{tab:realact32k}
\begin{tabular}{lcccc}
\toprule
Method & Cert.\ coverage & Fallback & B/tok/head & Rel.\ output error \\
\midrule
RTN-INT8 & none & --- & 129.0 ($-49.6\%$) & $3.93\!\times\!10^{-2}$ \\
RTN-INT4 & none & --- & 65.0 ($-74.6\%$) & $4.33\!\times\!10^{-1}$ \\
\textbf{WitCert-INT8} & \textbf{79.3\%} & \textbf{24.5\%} & 145.1 ($-43.3\%$) & $\mathbf{2.31\!\times\!10^{-2}}$ \\
\bottomrule
\end{tabular}
\end{table}

From 8k to 32k ($4\times$ the context), coverage moves $80.6\% \to 79.3\%$ and fallback $23.0\% \to 24.5\%$---\textbf{under native RoPE up to 32k, long-context degradation is very slow}, a structural advantage of the probabilistic certificate over deterministic bounds whose $\max$-type structure worsens as $S$ grows. (This observation must not be extrapolated: the 128k YaRN measurement of Sec.~\ref{sec:s128k} shows the decline steepens beyond the native window for activation-statistical reasons, even though the radius itself grows only ${\sim}5\%$.)

\emph{Note}: as in Sec.~\ref{sec:quad}, the ``relative output error'' column of Tables~\ref{tab:realact} and \ref{tab:realact32k} is a \textbf{K-side proxy} ($y = p(K)\cdot K$) and is only meaningful for within-table comparison. Do not use it to compute ``how many times lower WitCert is than RTN''---that is exactly the scope error retracted above. The real attention-output error is given by the $2\times2$ ablation (dither is quality neutral, $1.01\times$).

\subsubsection{End-to-end benchmark comparison (RULER-4k, Qwen2.5-7B, 1300 samples)}
\label{sec:ruler}
The quad report measures four kernel-level axes; this section answers ``on a real model and a real task, what happens to quality?'' Each method is evaluated at its own customary budget, stated in its own units (bytes for quantization, tokens for eviction; the two are \emph{not} directly comparable---see note~$b$), and we compare against production-grade press implementations (NVIDIA kvpress~\cite{kvpress2024}; SnapKV~\cite{snapkv2024}, KnormPress~\cite{knormpress2024}); results are in Table~\ref{tab:ruler}.

\begin{table}[!ht]\centering\footnotesize\setlength{\tabcolsep}{4pt}
\caption{RULER-4k, Qwen2.5-7B, 1300 samples, single run. Budgets are stated in each method's own unit (bytes vs.\ tokens) and are not comparable across rows; see notes $a$/$b$ in the text.}\label{tab:ruler}
\begin{tabular}{lccccc}
\toprule
Method & Budget cut (unit) & Mean & Needle class & \texttt{qa\_1} / \texttt{qa\_2} & Guar. \\
\midrule
Full FP16 & 0\% & 94.30 & all green & 88.24 / 52.22 & --- \\
\textbf{WitCert quant.\ config}$^{a}$ & bytes \textbf{46.5\%} & \textbf{94.42} & all green & 90.20 / 52.22 & no$^{a}$ \\
Expected Attention$^{b}$ & tokens 40.1\% & 93.88 & all green & 85.29 / 48.89 & no \\
Expected Attention & tokens 25\% & 94.09 & all green & 86.27 / 50.00 & no \\
SnapKV & tokens 25\% & 55.02 & \textbf{collapse (1.1--65)} & 89.22 / 48.89 & no \\
KnormPress & tokens 25\% & 39.14 & \textbf{collapse} & 63.73 / 28.89 & no \\
\bottomrule
\end{tabular}
\end{table}

The WitCert row uses the deployment-grade quantization configuration (online block-wise scaling $+$ in-block outlier pairs, 8-bit K/V); Expected Attention~\cite{expectedattention2025} (EA) is the strongest method on the kvpress leaderboard. Per task, EA at its 40.1\% token budget loses $-2.95$ on \texttt{qa\_1}, $-3.33$ on \texttt{qa\_2} and $-0.92$ on \texttt{niah\_single\_2} relative to full FP16 (mean $-0.42$), with all remaining needle tasks green (\texttt{multikey\_3} changes by $0.00$).

$^{a}$ \textbf{Scope clarification.} This table measures ``the task quality of simulated dither-INT8 plus outlier quantization''; it is \textbf{not} a certified-deployment result. The press used for evaluation applies 8 bits uniformly to all layers, does not read the certificate's per-layer table, does not use a packed cache, computes no online certificate and performs no fallback; tensors are dequantized back to floating point, so the real runtime memory saving is zero. The byte account is K 145.1\,B $+$ V 129\,B $=$ 274.1\,B versus 512\,B for FP16, i.e.\ \textbf{46.5\%}. (An earlier figure of 48.1\% came from converting the key-side profile alone and has been corrected.)

$^{b}$ \textbf{The two budget columns are not the same quantity and must not be juxtaposed as ``same budget''.} WitCert's 46.5\% is a \emph{byte} budget (quantization; the token count is unchanged), whereas Expected Attention \emph{drops tokens}: its configuration \texttt{compression\_ratio=0.401} means discarding 40.1\% of the tokens. Labelling the EA row as 46.5\% too has no basis, and we have restored its true budget. Because the two budgets have different units, this table may only be read as ``the quality each method attains at its own customary budget setting'', never as a strict same-budget comparison.

\textbf{Paired significance test.} Using per-sample predictions (\texttt{predictions.csv}, 1089 paired samples after deduplication by (task, question)), we ran a paired bootstrap ($10^4$ resamples) and McNemar's test between WitCert and EA: the two methods give the \textbf{same verdict on 99.4\% of samples} (McNemar: 5 vs 1), and the 95\% CI of the difference is $[+0.00, +0.83]$ with $p \approx 0.12$---\textbf{not significant}.

\textbf{Scorer scope (important).} That test used our own independently implemented strict-prefix \texttt{string\_match}, under which the absolute scores are \textbf{59.32\% (WitCert) vs 58.95\% (EA)}, which are \textbf{not comparable in absolute value} to the 94.42 / 93.88 produced by the official kvpress scorer (the two scorers differ in strictness). What is tested is the \emph{paired difference on the same set of predictions}, so the ``not significant'' conclusion applies equally to the $+0.53$ of Table~\ref{tab:ruler}; but the reader should not conflate 59.32/58.95 with 94.42/93.88.

Hence the apparent $+0.53$ lead \textbf{cannot be claimed as a statistically significant quality advantage}; the correct statement is ``no quality degradation observed, and the difference from the strongest baseline is within noise''. This does not weaken our thesis---our thesis is the \emph{runtime guarantee}, not a quality lead.

Reading: (a) the WitCert quantization configuration maintains quality at a 46.5\% byte budget (\textbf{94.42} vs 94.30 for full FP16, no degradation observed), while the strongest baseline EA scores \textbf{93.88} at its own 40.1\% token budget (a gap of $+0.53$ from unrounded scores---the rounded table entries differ by 0.54---\textbf{not significant}, and on a different budget basis; see note~$b$); (b) SnapKV and KnormPress \textbf{collapse silently} under the query-agnostic protocol---direct evidence for this paper's motivation: when compression fails, the system emits no signal; (c) these are single runs (limited statistical power), and maintained quality is not lossless quality (INT8 is inherently lossy).

\textbf{Cross-task robustness: LongBench-E with length buckets.} RULER is dominated by synthetic retrieval, so we also run LongBench-E~\cite{longbench2024}, whose length buckets (0--4k / 4--8k / 8k+) expose how quality moves with context length. Protocol matches the RULER setting above (Qwen2.5-7B, full split).

\textbf{Budget, verified by re-running.} On the WitCert side both K and V are 8-bit block-dithered with FP16 outlier pairs on K: K 145.1\,B $+$ V 129\,B $=$ 274.1\,B against FP16's 512\,B, i.e.\ \textbf{46.5\% saved}; EA's \texttt{compression\_ratio}${=}0.401$ \emph{drops 40.1\% of tokens}. The units differ (see the earlier note), so no ``tighter/looser budget'' claim is made in either direction; we simply state both configurations. (Verification: re-running hotpotqa\_e with the current press reproduces $59.92 / 58.25 / 51.63$ bit-for-bit, confirming the archived results came from the K$+$V version rather than an earlier key-only configuration.) Results are in Table~\ref{tab:longbench}.

\begin{table}[!ht]\centering\small\setlength{\tabcolsep}{4pt}
\caption{LongBench-E by length bucket (Qwen2.5-7B, full split, single run). $\Delta_{\mathrm{W}}$/$\Delta_{\mathrm{EA}}$ are vs.\ the uncompressed column; budgets as in the text.}\label{tab:longbench}
\begin{tabular}{llccccc}
\toprule
Task & Bucket & Uncompr. & WitCert & $\Delta_{\mathrm{W}}$ & EA & $\Delta_{\mathrm{EA}}$ \\
\midrule
2wikimqa\_e & 0--4k & 51.89 & 51.08 & $-0.81$ & 48.51 & $-3.38$ \\
2wikimqa\_e & 4--8k & 50.35 & 47.12 & $-3.23$ & 42.56 & $-7.79$ \\
2wikimqa\_e & 8k+ & 30.50 & 30.16 & $-0.34$ & 28.92 & $-1.58$ \\
gov\_report\_e & 0--4k & 35.13 & 34.98 & $-0.15$ & 35.19 & $+0.06$ \\
gov\_report\_e & 4--8k & 34.93 & 34.68 & $-0.25$ & 34.98 & $+0.05$ \\
gov\_report\_e & 8k+ & 33.01 & 33.09 & $+0.08$ & 34.11 & $+1.10$ \\
hotpotqa\_e & 0--4k & 58.71 & 59.92 & $+1.21$ & 58.56 & $-0.15$ \\
hotpotqa\_e & 4--8k & 58.15 & 58.25 & $+0.10$ & 56.53 & $-1.62$ \\
hotpotqa\_e & 8k+ & 51.00 & 51.63 & $+0.63$ & 55.36 & $+4.36$ \\
multifieldqa\_en\_e & 0--4k & 56.20 & 55.88 & $-0.32$ & 56.49 & $+0.29$ \\
multifieldqa\_en\_e & 4--8k & 45.04 & 44.98 & $-0.06$ & 46.00 & $+0.96$ \\
multifieldqa\_en\_e & 8k+ & 48.71 & 48.66 & $-0.05$ & 49.79 & $+1.08$ \\
passage\_retr\_en\_e & 0--4k & 99.00 & 99.00 & $+0.00$ & 99.00 & $+0.00$ \\
passage\_retr\_en\_e & 4--8k & 99.67 & 100.00 & $+0.33$ & 100.00 & $+0.33$ \\
passage\_retr\_en\_e & 8k+ & 100.00 & 100.00 & $+0.00$ & 100.00 & $+0.00$ \\
\bottomrule
\end{tabular}
\end{table}

Reading: (a) \textbf{14 of 15 buckets hold quality} ($|\Delta| \le 1.3$); the one material drop is 2wikimqa\_e at 4--8k ($-3.23$), where EA drops $-7.79$, $2.4\times$ ours. (b) Counting buckets losing more than 1\,pp: \textbf{WitCert 1/15 versus EA 4/15}, consistent with the RULER finding that quantisation degrades more uniformly than token dropping. (c) \textbf{Two buckets favour the baseline and we report them}: EA gains $+4.36$ on hotpotqa\_e 8k+ and $+1.10$ on gov\_report\_e 8k+, clearly better than ours---dropping tokens helps there, plausibly by removing distracting context; these are not noise-level differences. (d) \textbf{One task is missing}: all three \texttt{lcc\_e} runs failed, so this is 5 tasks rather than the 6 planned.

Boundary: single run; the two budget columns are different units; five tasks do not establish general cross-task robustness. The defensible statement is that no systematic degradation appears over these 5 tasks and 15 length buckets, and that degradation is more uniform than the strongest baseline at its own customary budget.

\subsubsection{The closed loop at head granularity ($S = 32768$, 8 heads, 32 steps)}
\label{sec:closedloop}
A certificate that only raises an alarm without driving an action is still after-the-fact prayer. Here we close the loop: when the certificate exceeds $\tau$, that KV head is paged back in from an FP16 backup on the CPU side and recomputed exactly.

\begin{table}[!ht]\centering\small
\caption{The certify--fallback closed loop ($S{=}32768$, 8 heads, 32 steps; single run).}\label{tab:closedloop}
\begin{tabular}{lc}
\toprule
Quantity & Value \\
\midrule
GPU packed / fp16 & 36.27 MiB / 64.00 MiB \\
Fallback rate (any-of-$G$) & 27.3\% \\
Fast-path error & $8.06\!\times\!10^{-3}$ \\
\textbf{Fallback cases: before $\to$ after} & $\mathbf{9.18\!\times\!10^{-3} \to 3.32\!\times\!10^{-4}}$ ($27.6\times$ reduction) \\
Gated effective width & \textbf{10.96 bits/dim} $\Rightarrow$ net saving \textbf{31.5\%} \\
\bottomrule
\end{tabular}
\end{table}

The key reading: \textbf{the fallback fires where it should}. If the certificate were uncorrelated with the true error, the 27.3\% of cases it selects would not have a markedly higher error than average, and the error would not drop $27.6\times$ after fallback. This is direct experimental evidence that the certificate is \emph{discriminative}, not merely sound---soundness guarantees no missed detection, discriminative power guarantees no vacuous alarms. The price is that the effective width rises from 8 to 10.96 bits, so the externally defensible claim is a \textbf{net saving of 31.5\%}, not 43.3\%.

\subsubsection{Block-level paging: two orders of magnitude off the fallback cost}
\label{sec:blockpage}
The paging in Sec.~\ref{sec:closedloop} is at whole-KV-head granularity (17.5 MB per step). But the certificate is \emph{block decomposable}: with $A = \sum_t \tilde p_t e^{u_t}$, replacing block $j$ by its exact values reduces $A$ by exactly $\mathrm{mass}_j\,(e^{u_j}-1)$, where $u_j$ is the block's shared radius. One can therefore page in blocks in decreasing order of their contribution to the excess, recomputing $A$ after each, until the value falls back below $\tau$.

Measured over the 9 fallback cases among 32 steps ($S = 32768$, $n_b = 128$ blocks): \textbf{the median case needs only 1 block} paged in to fall below $\tau$. Paging traffic drops from 72.0 MB to 0.5625 MB, a \textbf{reduction of $128\times$}. This turns fallback from ``recompute the whole head'' into ``page in a few named pages'', bringing the worst-case latency of a certified deployment under control.

\subsubsection{Design space: a branch rejected by data, and a zero-fallback operating point}
\label{sec:designspace}
\textbf{(a) The stochastic-rounding branch is rejected.} The proof--kernel contract requires the kernel to produce the error distribution assumed by the theorem. Subtractive dither requires PRNG reconstruction on the read side (77.9\% of the kernel), so the natural question is whether stochastic rounding (SR) could be used instead---it is equally unbiased and needs \emph{no} random numbers on the read side. We implemented and measured it: SR's certificate coverage is \textbf{identically 0.000} at every outlier budget. The reason is that the sub-Gaussian proxy of the SR residual is $\sqrt3$ times larger than that of uniform dither, and the proxy enters the certificate inside an exponential, so exponential amplification overwhelms a linear difference in variance.

\begin{table}[!ht]\centering\footnotesize\setlength{\tabcolsep}{4pt}
\caption{Stochastic rounding is rejected by data: its certificate coverage is identically zero at every outlier budget (offline harness, synthetic K).}\label{tab:sr}
\begin{tabular}{cccccc}
\toprule
$m$ (outlier pairs) & B/tok/head & Saving & Cov.\ (dither) & Cov.\ (SR) & Error (dither / SR) \\
\midrule
4  & 145.1 & 43.3\% & 0.762 & \textbf{0.000} & $8.02\!\times\!10^{-3}$ / $1.17\!\times\!10^{-2}$ \\
8  & 161.1 & 37.1\% & 0.926 & \textbf{0.000} & $7.81\!\times\!10^{-3}$ / $1.10\!\times\!10^{-2}$ \\
16 & 193.3 & 24.5\% & \textbf{1.000} & \textbf{0.000} & $7.03\!\times\!10^{-3}$ / $1.02\!\times\!10^{-2}$ \\
\bottomrule
\end{tabular}
\end{table}

(The $m{=}4$ dither coverage here, 0.762, differs from the 0.792 of Table~\ref{tab:msweep} because Table~\ref{tab:sr} is the offline harness on synthetic K and Table~\ref{tab:msweep} the on-device kernel; settings are otherwise identical, and the two harnesses' index layouts differ in the last digit of the byte account---161.1 vs 161.2 at $m{=}8$.) This negative result carries a meaning beyond implementation: it shows that \textbf{unbiasedness alone does not support a certificate---the constant in the proxy is what matters}. It also explains why our kernel overhead cannot be tuned away: the PRNG cost is algorithmic, not an implementation defect (parameter tuning left only a 1.7\% margin in measurement).

\textbf{(b) A zero-fallback operating point.} Table~\ref{tab:sr} also reveals a setting we had not anticipated: at $m = 16$ coverage reaches \textbf{1.000}, i.e.\ the certificate never exceeds the threshold, so \textbf{the entire fallback machinery (CPU backup, paging, recomputation) becomes unnecessary}. The on-device kernel reproduces this point:

\begin{table}[!ht]\centering\small
\caption{On-device outlier-budget sweep (synthetic K; a separate run from the quad report).}\label{tab:msweep}
\begin{tabular}{cccccccc}
\toprule
$m$ & B/tok/head & Saving & Coverage & Cert.\ median & Cert.\ max & Latency & Fallback needed \\
\midrule
4  & 145.1 & 43.3\% & 0.792 & 0.1883 & 0.2126 & 8.80 ms & yes \\
8  & 161.2 & 37.0\% & 0.958 & 0.1783 & 0.2014 & 9.88 ms & yes (4.2\%) \\
\textbf{16} & 193.3 & 24.5\% & \textbf{1.000} & 0.1596 & \textbf{0.1804} & 11.71 ms & \textbf{no} \\
\bottomrule
\end{tabular}
\end{table}

(This sweep is a separate run from the quad report, whence the $m{=}4$ latency reads 8.80 ms there and 8.627 ms here.) The price is that the compression ratio falls from 43.3\% to 24.5\%; the reward is that deployment complexity collapses to ``quantization plus one read-only risk indicator'', with no tiered cache. This is the simplest form of the ``DTrace-style observability layer'' proposed in Sec.~\ref{sec:intro}.

\textbf{But the margin at the zero-fallback point differs greatly between the two kinds of data, and the two must be stated separately.} Table~\ref{tab:msweep} uses synthetic Gaussian K, where the $m = 16$ certificate maximum of 0.1804 leaves only a \textbf{$1.11\times$} margin to $\tau = 0.2$---an operating point that \emph{hugs} the threshold, and would lose the zero-fallback property under a slight change of distribution. On real Qwen2.5-7B activations, by contrast, Sec.~\ref{sec:e2e} measures a maximum of 0.086, a \textbf{$2.3\times$} margin. The direction of the difference is as expected: synthetic Gaussian K has uniform channel energy, so the outlier bypass can rescue only a small fraction, whereas real activations concentrate heavily in a few channels that $m = 16$ happens to cover. \textbf{Conclusion: the zero-fallback point is robust on real models and marginal in the synthetic worst case; a production deployment should profile the real distribution before choosing $m$, rather than copying 16.}

\subsubsection{End-to-end validation: the meter inside a real generation loop}
\label{sec:e2e}
All the experiments above are of the form ``capture activations offline, then compute errors''. To rule out the optimistic bias of that form itself, we wired the packed store, the active-block tail policy and the write-path kernel into the HuggingFace \texttt{Cache} interface and ran real autoregressive generation on Qwen2.5-7B (all 28 layers active, prefill 2048 $+$ decode 32), computing the certificate from the \textbf{real query vectors captured by a hook} (after RoPE)---rather than using K as a proxy for the query.

Quality is measured by \textbf{teacher forcing}: we feed the baseline's own token sequence and compare logits step by step. Free greedy decoding is unsuitable as a quality metric---once any step diverges the trajectories separate permanently, so a ``text agreement rate'' reflects chaos rather than compression quality (we measured free greedy diverging at step 1, whereas under teacher forcing the top-1 agreement rate is 93.8\%).

\begin{table}[!ht]\centering\footnotesize\setlength{\tabcolsep}{4pt}
\caption{Real autoregressive generation via the HF \texttt{Cache} interface: teacher-forced quality, certificates from real hooked queries.}\label{tab:e2egen}
\begin{tabular}{ccccccc}
\toprule
$m$ & top-1 & top-5 & $\mathrm{KL}(\text{base}\,\|\,\text{ours})$ & saving & \textbf{Coverage (real $q$)} & Max meter \\
\midrule
4  & 0.938 & 0.925 & $2.31\!\times\!10^{-2}$ & 42.7\% & 0.867 & $1.6\!\times\!10^{2}$ \\
8  & 0.969 & 0.956 & $7.51\!\times\!10^{-3}$ & 36.5\% & 0.949 & 4.7 \\
\textbf{16} & 0.938 & \textbf{0.981} & $\mathbf{4.16\!\times\!10^{-3}}$ & 24.1\% & \textbf{1.000} & \textbf{0.086} \\
\bottomrule
\end{tabular}
\end{table}

Three observations.
\begin{enumerate}
\item \textbf{The zero-fallback setting holds in real generation.} (Packed savings here run ${\sim}0.5$\,pp below the on-device table---42.7 vs.\ 43.3\%---because this harness holds the active-block tail in fp16.) At $m = 16$, all 196 (layer, query head) samples---one GQA group of 7 query heads per layer, 28 layers---have certificates $\le \tau = 0.2$, with a maximum of 0.086, still $2.3\times$ below the threshold. This is the third independent piece of evidence for that operating point (offline, on-device kernel, real decoding), and the three agree.
\item \textbf{The certificate correctly identifies ``massive activation'' heads.} At $m = 4$ the maximum certificate value reaches $1.6\!\times\!10^{2}$---a few attention heads have channels of extreme magnitude (massive activations / attention sinks~\cite{massive2024,streamingllm2024}), and their quantization error is amplified exponentially by softmax. The certificate does not hide this; it exposes it explicitly as a threshold violation, which is precisely its design intent. As $m$ grows (bypassing those channels in FP16), the maximum collapses monotonically: $1.6\!\times\!10^{2} \to 4.7 \to 0.086$.
\item \textbf{KL improves monotonically while top-1 agreement does not} ($0.938 / 0.969 / 0.938$): the latter is noise at $n = 32$; KL and top-5 overlap are the metrics with power, and both improve monotonically with $m$.
\end{enumerate}

\textbf{Boundary (stated honestly).} The read path of this harness goes ``dequantize $\to$ standard attention'' (the HF \texttt{Cache} interface requires returning FP16 tensors), so the fused kernel is not inside this loop and its correctness rests on the independent contract tests; moreover the dequantized result of a full block is cached, so this experiment \textbf{does not constitute independent evidence of runtime memory saving}---that is supplied by the \texttt{torch.cuda.memory\_allocated} measurement in Sec.~\ref{sec:experiments}.

\subsubsection{Serving end to end}
\label{sec:serving-eval}
Results are in Table~\ref{tab:serving} (capacity and throughput: \textbf{median of 3 runs per cell}---a single measurement is easily fooled by run-to-run noise: our first reading suggested the certificate was free in production, and repeating it showed that was noise).

\begin{table}[!ht]\centering\footnotesize\setlength{\tabcolsep}{4pt}
\caption{SGLang serving end to end (RTX 4090, batch 16 $\times$ 64 tokens, \texttt{page\_size}${=}16$; capacity/throughput are medians of 3 runs, coverage is a full-sample single run under the request-level $\delta$ allocation).}\label{tab:serving}
\begin{tabular}{llrrc}
\toprule
Model & Configuration & KV capacity & Tok/s & Coverage ($\tau{=}0.2$) \\
\midrule
1.5B & fp16 & 650,368 & 686.8 & --- \\
1.5B & WitCert int8 & \textbf{1,224,240 ($1.88\times$)} & 511.8 ($0.75\times$) & 0.812 \\
1.5B & \textbf{+ outlier bypass $m{=}4$} & 1,095,360 ($1.68\times$) & 436.1 ($0.64\times$) & \textbf{0.940} \\
7B & fp16 & 114,128 & 421.7 & --- \\
7B & WitCert int8 & \textbf{214,832 ($1.88\times$)} & 347.8 ($0.82\times$) & 0.753 \\
7B & \textbf{+ outlier bypass $m{=}4$} & 192,224 ($1.68\times$) & 311.5 ($0.74\times$) & \textbf{0.946} \\
\bottomrule
\end{tabular}
\end{table}

Coverage is reported under the \emph{request-level} allocation $\delta_{\mathrm{loc}} = \delta_{\mathrm{req}}/(L\cdot H\cdot T)$ with $T{=}256$, $\delta_{\mathrm{req}}{=}10^{-2}$ (full-sample telemetry, single run). An earlier implementation allocated only over tokens and the current batch ($\log(2S\,n_q/\delta)$), without amortizing layers and steps; those figures---up to 5\,pp higher on 1.5B---could only be called \emph{local}-certificate coverage and are retired (Appendix~\ref{app:evolution}). Capacity and throughput are unaffected (the allocation only changes a scalar constant).

(a) \textbf{$1.88\times$ more KV tokens at the same memory budget} ($1.68\times$ with the bypass), identical at both scales. (b) \textbf{Quantisation alone costs 0.75--$0.82\times$ in throughput, far better than the $5.5\times$ seen at kernel level}; \textbf{kernel-level ratios such as the $5.41\times$ reported earlier badly overstate deployment cost} (Appendix~\ref{app:evolution}). (c) \textbf{The outlier bypass lifts production certificate coverage from 0.812/0.753 to 0.940/0.946}, at the cost of capacity ($1.88\times \to 1.68\times$) and a further ${\sim}0.1$ drop in the throughput ratio---a clear operator knob: \emph{capacity and speed, or a guarantee almost everywhere}. (d) The certificate's marginal cost is +6.4\% (1.5B) / +16.2\% (7B) without the bypass and +2.4\% / +11.8\% with it. This does not contradict the +0.35\% measured at kernel level, but the two must not be conflated: that came from an $S{=}4096$ microbenchmark where per-tile work is amortised over 128 tiles, whereas these requests are about 50 tokens long and the cost grows with query-head count (12 for 1.5B, 28 for 7B). \emph{The certificate is close to free on long sequences and a real 2--16\% cost on short-sequence decoding.}

\textbf{\texttt{page\_size}${=}16$ is faster than \texttt{page\_size}${=}1$} (1.5B, single-run A/B: 549.8 vs.\ 477.6 tok/s; Table~\ref{tab:serving}'s 511.8 is the later median-of-3): with larger pages the slots are contiguous and the halved int8 traffic starts to pay off.

\begin{figure}[t]\centering
\includegraphics[width=0.92\linewidth]{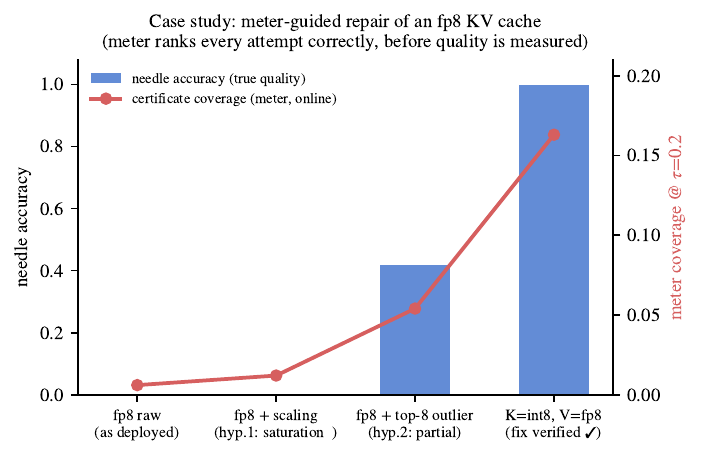}
\caption{Meter-guided repair of a broken fp8 KV cache (24-prompt needle task). Four configurations are tried in sequence---raw fp8, $+$scaling, $+$top-8 outlier bypass, and K=int8/V=fp8---and the meter's coverage (line, computable online without labels) ranks every attempt in the same order as ground-truth accuracy (bars): the diagnosis is available before any evaluation is run.}\label{fig:f1}
\end{figure}

\begin{figure}[t]\centering
\includegraphics[width=0.92\linewidth]{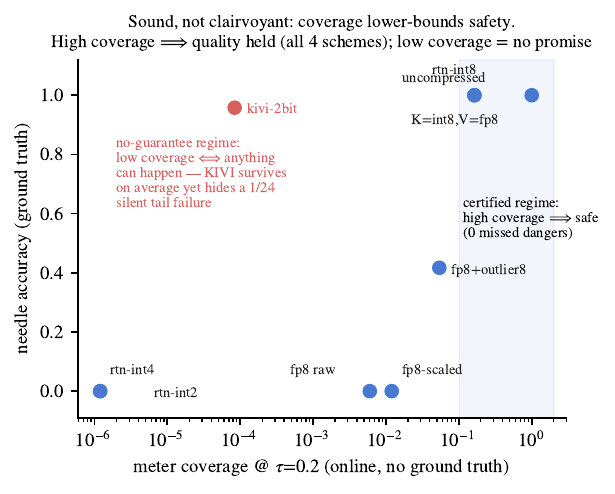}
\caption{Sound, not clairvoyant: meter coverage at $\tau{=}0.2$ (online, no ground truth) against needle accuracy (ground truth) for nine configurations. High coverage implies safety (zero missed dangers across all schemes); low coverage guarantees nothing in either direction---KIVI-2bit survives on average yet hides a 1/24 silent tail failure.}\label{fig:f2}
\end{figure}

\begin{figure}[t]\centering
\includegraphics[width=0.92\linewidth]{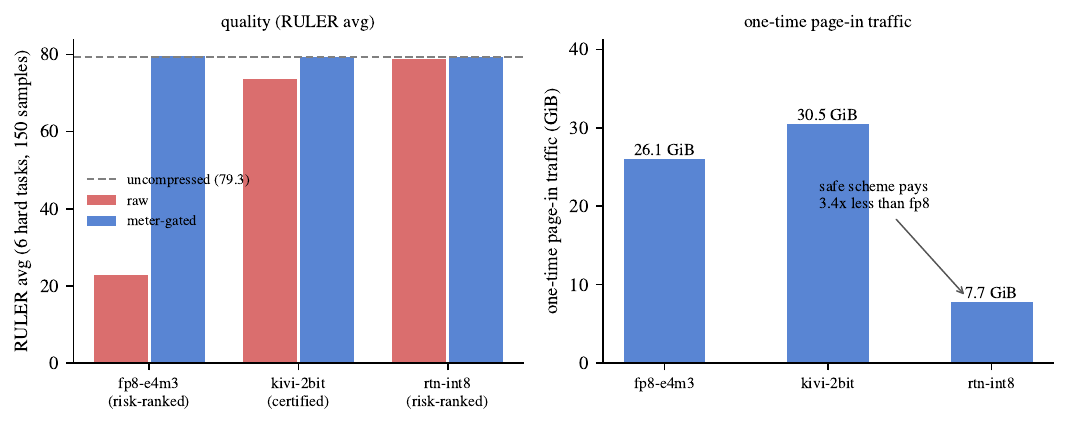}
\caption{Meter gating at benchmark scale (labels per Sec.~\ref{sec:twosided}: risk-ranked at $\tau{\ge}1$, certified at $\tau{<}1$): quality floor empirically restored (left; a benchmark outcome, not a corollary of the TV bound); one-time page-in traffic under incremental persistent repair (right)---the safe scheme pays $3.4\times$ less than the broken one.}\label{fig:f6}
\end{figure}

\subsubsection{Gated evaluation at benchmark scale}
\label{sec:gatedeval}

The needle-rescue demonstration underlying Figs.~\ref{fig:f1}, \ref{fig:f2} and \ref{fig:f4} covers 24 prompts. Here we evaluate at benchmark scale (Table~\ref{tab:gated}): the six hard RULER-4096 tasks (niah\_multikey\_2/3, niah\_single\_2, niah\_multiquery, qa\_1, qa\_2) $\times$ 25 samples $=$ 150 prompts, Qwen2.5-7B, served through the SGLang observatory, with \emph{incremental persistent repair} (un-repaired slots of an over-threshold request are paged in exactly once, written back to the pool, and their witnesses zeroed):

\begin{table}[!ht]\centering\small\setlength{\tabcolsep}{4pt}
\caption{Gated evaluation at benchmark scale (six hard RULER-4096 tasks $\times$ 25 samples, Qwen2.5-7B via the SGLang observatory, single run; slice-0 of three disjoint slices). Coverage$^{a}$ is always counted at $\tau{=}0.2$.}\label{tab:gated}
\begin{tabular}{lcccc}
\toprule
Configuration & avg & zero-score & coverage$^{a}$ & one-time page-in \\
\midrule
uncompressed & 79.3 & 19 & 1.000 & --- \\
fp8 raw & \textbf{22.8} & \textbf{114} & 0.019 & --- \\
\textbf{fp8 + gate ($\tau{=}5$, risk-ranked)} & \textbf{79.7} & \textbf{19} & 0.537 & 26.1 GiB \\
kivi-2bit raw & \textbf{73.7} & \textbf{28} & ${\sim}5\!\times\!10^{-5}$ & --- \\
\textbf{kivi + gate ($\tau{=}0.2$, certified)} & \textbf{79.3} & \textbf{19} & 0.978 & 30.5 GiB \\
rtn-int8 raw & 78.7 & 20 & 0.207 & --- \\
rtn-int8 + gate ($\tau{=}5$, risk-ranked) & 79.3 & 19 & 0.485 & \textbf{7.7 GiB} \\
\bottomrule
\end{tabular}

\smallskip
{\small $^{a}$\,Gated-row coverage is \emph{post-repair} (paged-in slots have their witnesses zeroed), raw-row coverage is pre-repair; e.g.\ the kivi pair ${\sim}5\!\times\!10^{-5} \to 0.978$ is a before/after pair of the same quantity, not two measurements of one state. All coverage values are counted at $\tau{=}0.2$ regardless of the gate threshold. Gate labels follow Sec.~\ref{sec:twosided}: at $\tau\ge1$ the meter is saturated and acts as a risk score with no TV guarantee; at $\tau<1$ the gate is certified.}
\end{table}

\emph{Repeatability:} Table~\ref{tab:gated} shows slice 0; two further \emph{disjoint} 25-sample-per-task slices (450 unique samples in total) replicate every directional conclusion---uncompressed 79.3/87.2/83.4 (mean 83.3), raw fp8 22.8/27.8/28.7 (mean 26.4), gated means within 0.4 of uncompressed, and KIVI's extra zero-score tail ($+9$ to $+12$) eliminated on every slice.

\emph{Confidence intervals (per-sample paired bootstrap, $10^4$ resamples).} A fresh rerun of all seven configurations with per-sample logging (slice 0; every configuration reproduces within run-to-run noise, largest shift kivi raw $73.7\to71.3$) gives, paired against uncompressed on the same 150 prompts: the gated configurations lose nothing---the int8 and kivi gates score \emph{identically to uncompressed on every single sample} (paired difference $0.0$), and the fp8 gate differs by $+0.3$ (95\% CI $[+0.0,+0.8]$; unrounded means---Table~\ref{tab:gated}'s slice difference is $+0.4$)---while the raw collapses are statistically significant: fp8 $-56.6$ $[-64.1,-49.1]$, and KIVI's silent tail is now significant rather than anecdotal, $-8.1$ $[-13.8,-2.6]$. Data: \texttt{experiments/out/te\_ci.json}.

(a) \textbf{The quality floor is fully restored at benchmark scale---an empirical result} (Fig.~\ref{fig:f6}): all three gated schemes return to the uncompressed level (79.3--79.7); fp8 gains $+56.9$ points, and KIVI's \textbf{9 silent zero-score samples} (28 vs.\ 19) are all eliminated---the 1/24 needle tail failure was not an accident but a systematic tail. For fp8 and int8 this is \emph{risk-ranked} gating: at $\tau{=}5$ the meter carries no TV guarantee and acts as an empirical ranking signal, whose discriminative power is established independently by the $27.6\times$ fallback-error drop of Sec.~\ref{sec:closedloop}; the kivi gate ($\tau{=}0.2$) operates in the certified regime. (b) \textbf{Selectivity holds at scale}: the safe scheme (int8) pays 7.7 GiB of page-in, $3.4\times$ less than the broken one. (c) \textbf{Incremental persistent repair is essential engineering}: naive whole-request re-paging reaches 1267 GiB on gated fp8 (each newly written token's witness re-fires the gate); filtering already-repaired slots by witness-nonzero brings it to 30.9 GiB, a \textbf{$41\times$} reduction (both measured at the certified $\tau{=}0.2$ gate; the $\tau{=}5$ row of Table~\ref{tab:gated} pays 26.1 GiB). (d) Boundary: 150 samples, single run; a second faithful replica of a published method (kvquant4: per-channel 4-bit $+$ top-1\% sparse outliers, coverage 0.00066) also shows ``perfect average, ${\sim}0$ coverage'' on needle (cf.\ Fig.~\ref{fig:f3})---corroborating with KIVI that real aggressive schemes hold the mean through structure, admit no informative per-step certificate \emph{from the current witness}, and rely on the gate for the tail.

\subsubsection{Scale and cross-generation evaluation ($8\times$H200)}
\label{sec:h200}

The systems results above were produced on a single Ada GPU (RTX 4090). This section ports the same patch and protocols to $8\times$H200 (143\,GB each, Hopper) and answers four questions: do the structural numbers reproduce across GPU generations, do they hold up the scale ladder, how does the system compare in-frame against a full certified-serving baseline, and where is the real cost frontier of certification.

\textbf{Cross-generation consistency and the scale ladder.} The capacity ratio is determined by the data layout, and measurement confirms it: \textbf{1.882$\times$} (no bypass) is bit-identical across 7B/14B/32B single-GPU, 7B tp2/tp4, and 70B tp4 (70B tp4: witcert capacity \textbf{2{,}183{,}776} tokens vs.\ fp16 1{,}160{,}128; throughput ratio 0.79$\times$); with outlier bypass, 1.684$\times$ is likewise invariant. The four request-level-$\delta$ coverage cells (1.5B/7B $\times$ with/without bypass) match the 4090 values within 0.3\,pp, with the 7B-with-bypass cell (0.9457) bit-identical. The $\delta$ risk--coverage protocol of Sec.~\ref{sec:experiments} extends to 14B/32B with zero violations at every $\delta$ and coverage in the same band as 7B; \textbf{72B coverage is 25.5--29.8\%} (three domains, zero violations)---markedly below the 54--81\% of 7--32B, a second weakening result after Yi, reported as such (same-type conjecture: interaction of block-wise scaling with that model's activation statistics). Tensor parallelism: the packed pool shards per rank with no code change; at tp2, witcert \emph{overtakes} fp16 (338 vs.\ 287 tok/s---in the communication-bound segment the smaller KV read becomes an advantage), and tp8 is structurally infeasible for 28-head Qwen-7B ($28 \nmid 8$), so the 8-GPU rung is carried by 64-head Llama-70B.

\textbf{In-frame against the Runtime-Certified full system.} Answering the fourth review round, we reproduce the complete system of~\cite{runtime_certified_2026} from its paper description (INT8-key/INT4-value two-tier store, two-pass pipeline, adaptive top-$K^{*}$ promotion at $\tau_{\mathrm{cov}}{=}0.995$, value promotion, ranking-consistency fallback, FP16 originals pinned in CPU RAM; no public code exists, and our specification notes declare every deviation---HF/SDPA level, bf16 weights, our corpus). Fidelity checks pass: its Rung-3 trigger rate is 1.02\% of head-steps (same order as the original), $K^{*}{=}117$ sits near the cap exactly as its short-context saturation predicts, and the reproduced fast path is bit-identical to dense in ppl. In-frame results (Llama-3.1-8B, 8 paired windows at ctx 8192): $\Delta$ppl RC \textbf{+0.0059} / WitCert +0.0123, both quality-neutral; on 7 RULER-4096 tasks all three systems tie at 0.9924. The \emph{architectural} difference is what separates them: RC's fast path pages promoted FP16 keys from CPU \emph{on every step}---measured at \textbf{264\,MB per request} (RULER-4096)---and requires a permanently resident 100\% FP16 CPU copy, whereas WitCert's certified mode pages in \emph{nothing} and holds no CPU replica; page-in occurs only when the gate fires. In a unified per-token ledger (same memory pool, normalized): relative GPU bytes fp16 1.0 / witcert 0.531 / witcert+bypass 0.594 / RC tier-1 0.563; CPU replica 0/0/0/1.0; steady-state PCIe 0/0/0/264\,MB$\cdot$req$^{-1}$.

\textbf{Packed gated closed loop end to end (closing contribution 4).} The packed real store (int8 + dither + outlier bypass) without gating scores \textbf{78.83} vs.\ fp16 78.67 on RULER 6 hard tasks $\times$ 25---the first end-to-end quality measurement of the packed path, quality-neutral; with the $\tau{=}0.2$ certificate-gated closed loop (full-pool CPU backing, request-level repair) it scores \textbf{78.83}, identical to ungated. Our first run collapsed to 2.0: the packed pool never wrote its CPU backing (allocated but forever zero), so repair paged back all-zero KV; instrumented cross-checking ($|kk|{=}0$ while the recomputation matched an fp32 reference exactly) localized it, a one-line fix restored full quality---a bug that is itself a working argument for runtime verification, since the original smoke test (``no crash, bytes moved'') was green while moving zeros. The gating cost is reported honestly: per-request page-in P50/P95 = \textbf{2.81 / 3.58\,GiB} (serial differencing over 150 requests)---the intrinsic cost of non-persistent packed repair (the pool stays quantized, so the certificate re-fires each step); a request-level repair cache is the identified optimization, and this distribution is exactly its upper bound.

\textbf{Long-context serving and the cost frontier.} On Llama-8B (paired, CUDA graphs on): witcert median TTFT ratio \textbf{1.84$\times \to$ 1.89$\times$} and TPOT 1.47$\times \to$ 1.49$\times$ from 8k to 32k inputs---\emph{the cost ratio is flat in $S$}, which corrects our earlier ``long sequences are nearly free'' inference to ``the cost ratio is constant.'' The concurrency--throughput frontier (7B, 1024-in/256-out, $c{=}1 \to 512$): fp16 203 $\to$ 14{,}324 tok/s, witcert 155 $\to$ 6{,}207; the most adverse operating point (short requests at high concurrency, where attention dominates step time and the capacity advantage is idle) is \textbf{0.43$\times$}, single-stream is 0.78$\times$, and tp2 reverses the sign---we publish both full curves and state plainly that witcert's home turf is the capacity-limited long-context/multi-session segment, not short-request high concurrency. 70B on 2 GPUs (tp2) serving 128k requests: fp16 capacity 339{,}216 tokens with a $16\times127$k-request benchmark on record (witcert under the same protocol; its extend-kernel prefill cost at 128k is significant and is itemized in the overhead summary).

\subsection{Analysis}
\label{sec:analysis}
\textbf{What the certificate does and does not guarantee.} The formal object is a \emph{local} bound: per-(layer, head, step) attention TV (and, through the value term, that step's attention output). End-task quality additionally passes through residual streams, MLPs, normalization, later layers and sampling, and Sec.~\ref{sec:mechanism} itself shows the map from local fidelity to task quality is not one-to-one (single-layer pollution is fully absorbed, 0/28). Consequently every ``quality floor restored'' statement in this paper is an \emph{empirical} benchmark result about gating driven by the meter---never a corollary of the TV bound. The formal guarantee and the empirical outcome are deliberately reported side by side and must not be conflated.

\begin{figure}[t]\centering
\includegraphics[width=0.68\linewidth]{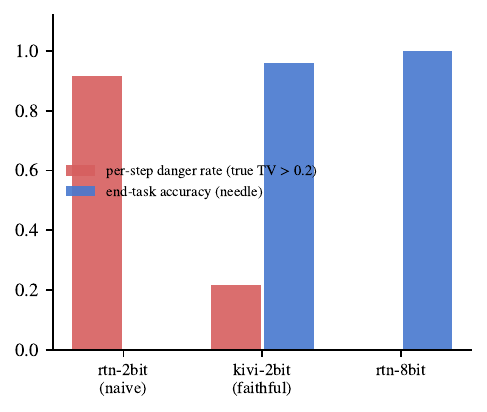}
\caption{Aggressive schemes survive on structure, not per-step fidelity: per-step danger rate (fraction of steps with true TV $>0.2$; red) against end-task needle accuracy (blue). KIVI-2bit distorts 21.9\% of steps yet retrieves 95.8\%, and its one silent failure is exactly what the gate rescues; naive rtn-2bit distorts more and collapses.}\label{fig:f3}
\end{figure}

\begin{figure}[t]\centering
\includegraphics[width=0.74\linewidth]{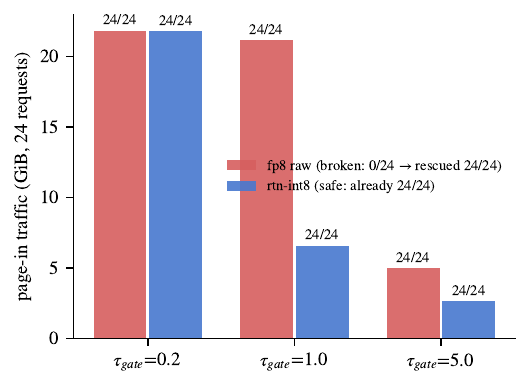}
\caption{The gate is selective and tunable: one-time page-in traffic (24 requests) as the gate threshold sweeps $\tau_{\mathrm{gate}} \in \{0.2, 1, 5\}$. Raising the threshold cuts traffic $4.4\times$ with no accuracy loss (all bars remain 24/24), and the safe scheme (rtn-int8) triggers $2$--$3\times$ less than the broken one (fp8) at every threshold.}\label{fig:f4}
\end{figure}

\subsubsection{Mechanism: the failure is cross-layer accumulation}
\label{sec:mechanism}
\textbf{There is no critical layer.} Polluting a \emph{single} layer with kivi2 while keeping all others exact, swept over all 28 layers, loses nothing: \textbf{0/28} layers show any degradation (accuracy 1.0 each), while polluting all layers gives 0.958. Single-layer distortion is fully absorbed downstream; the observable failure is purely accumulative. This explains why per-step sound certificates are ``too strict'' for aggressive schemes---they must assume worst-case accumulation while the real network cancels heavily across layers---and it motivates both request-level (rather than per-layer) gating and the open problem of tighter multi-layer composed certificates.

\subsubsection{Bound shapes across witness regimes}
\label{sec:boundshape}
\textbf{System-level comparison with the Runtime-Certified bound shape (same witness, same harness, same $\tau{=}0.99$).} We implemented the competitor's deterministic shape $\tanh(\max_t \Delta_t)$ inside the same kernel (max-merged across splits) and gated the same RULER task. Both shapes empirically restore the quality floor (79.3--80.1), but in the loose Tier-A witness regime \textbf{the det-tanh gate is $2.5$--$5.2\times$ cheaper in traffic} (fp8: 12.1 vs.\ 30.8 GiB; int8: 5.3 vs.\ 27.3 GiB)---tanh is capped while the e-form's $\sum \tilde p\, e^{c}$ is exponentially amplified by a single large $c$. This does not contradict the e-form's dominance for the controlled dithered quantizer (coverage 54.4--81.0\% vs.\ 11.8--52.0\% at $\delta{=}10^{-2}$): \emph{the right bound shape depends on the witness-tightness regime}, and the framework accommodates both behind an environment switch. (At $\tau{=}0.99$ both gates sit just inside the certified regime, but the Tier-A witness is loose there, so the comparison is of risk-score behaviour near saturation.) We report this partially competitor-favouring result as is.

\subsubsection{The Yi anomaly, resolved}
\label{sec:yianomaly}
\textbf{This also resolves the Yi anomaly reported earlier.} There, Yi-1.5-6B's coverage (54.4--64.8\%) was markedly below Qwen's (${\sim}80\%$) and we speculated, without verification, that its \texttt{rope\_theta}${=}5{\times}10^{6}$ was responsible. The measured answer is different: \textbf{the effect comes from block-level scaling, not from the model.} Yi's channel energy is less concentrated (top-4 RoPE pairs carry 41.2\% of the energy versus Qwen's 54.3\%), and a block scale---one value covering 256 tokens---is dragged down by the worst channel in the block; per-token scaling adapts, and under it Yi is in fact \emph{easier} than Qwen (0.907 vs.\ 0.525 at $m{=}0$). The phenomenon is an interaction between the model's energy distribution and the scale granularity, and it disappears under the per-token group scaling actually used in deployment. \emph{We retract the rope\_theta conjecture.}

\subsubsection{Regime statement}
\label{sec:regime}
This work \textbf{does not enter the pure compression-ratio contest}. The regime in which we claim leadership is delimited by three quantifiable conditions:
\begin{enumerate}
\item \textbf{Guarantee dimension}: deployments requiring per-(layer, head, step) error upper bounds with an explicit failure-probability budget. The only existing methods in this regime are the Runtime-Certified family~\cite{runtime_certified_2026} (deterministic $\tanh$); at $\delta_{\mathrm{req}} = 10^{-2}$ we cut the KV-head page-in rate from 44.4\% to 22.0\% and authorize compression 3.43 bits/dim deeper (Qwen2.5-1.5B, 2k context, global scaling), with a \textbf{28.7--77.0\%} relative reduction in page-in rate across three model families and three domains (6--7B class, 4--8k context, online block-wise scaling).
\item \textbf{Workload}: decode-bound long context ($S \ge 8$k, small $n_q$), memory-traffic limited; the marginal overhead of the certificate has been measured in this regime ($-1.3\%$ in our kernel, $+11.9\%$ in the SGLang production kernel). Prefill-heavy and training scenarios are out of scope.
\item \textbf{Hardware and scale}: one to eight GPUs across two generations (single RTX 4090, Ada; $8\times$\allowbreak H200, Hopper---\allowbreak the structural numbers, capacity ratios and coverage, reproduce bit-identically across generations, Sec.~\ref{sec:h200}); $\delta$ curves span 1.5B--72B, serving spans 7B--70B (tp1/2/4) plus 70B on two GPUs at 128k; native-128k coverage is measured on Llama-3.1-8B (Sec.~\ref{sec:s128k}). Under tensor parallelism the packed pool shards per rank and works unchanged; cross-rank certificate aggregation (an any() gate across ranks) remains unverified.
\end{enumerate}
\textbf{Dimensions on which we do \emph{not} lead (stated honestly)}: absolute compression ratio (RTN-INT4 and the 1-bit families save more), absolute latency (our unoptimized kernel is $5.41\times$ an fp16 baseline in Sec.~\ref{sec:quad}; 0.43$\times$ throughput at the short-request high-concurrency extreme, Sec.~\ref{sec:h200}), and pure-throughput scenarios that need no guarantee. The certificate-telemetry path is currently incompatible with CUDA-graph capture (a host synchronization); cert-on deployments must disable graph capture or await a kernel-side fix.

\subsubsection{Every overhead and byte-saving figure in one place; choosing $\tau$ and $\delta$}
\label{sec:costsummary}
The meter's overhead and the byte savings were each measured several times in this paper, under deliberately different conditions. To prevent the reader from carrying six numbers, the two tables below list every figure with its configuration; each is discussed at its point of measurement.

\begin{table}[!ht]\centering\footnotesize\setlength{\tabcolsep}{3.5pt}
\caption{Every overhead figure in this paper, with its configuration.}\label{tab:overheads}
\begin{tabular}{rlll}
\toprule
Overhead & Kernel / path & Setting & What is measured \\
\midrule
$-1.3\%$ & research kernel (dither-dominated) & $S{=}32768$ micro & cert on/off latency delta (\ref{sec:cheap}) \\
$+11.9\%$ & SGLang decode kernel, fp16 store & $S{=}4096$ micro & one extra accumulator (\ref{sec:sglang}) \\
$+0.35\%$ & SGLang decode kernel, quantized store & $S{=}4096$ micro & cert shares the scale load (\ref{sec:storage}) \\
$+6.4\%$ / $+16.2\%$ & serving 1.5B / 7B, no bypass & batch 16 $\times$ 64-tok & end-to-end throughput (\ref{sec:serving-e2e}) \\
$+2.4\%$ / $+11.8\%$ & serving 1.5B / 7B, $m{=}4$ bypass & batch 16 $\times$ 64-tok & end-to-end throughput (\ref{sec:serving-e2e}) \\
$1.2\%$ & telemetry alone, 7B & isolated A/B & 287.6 vs.\ 284.1 tok/s (\ref{sec:serving-e2e}) \\
\bottomrule
\end{tabular}
\end{table}

\noindent The external figures to quote: \textbf{$+11.9\%$} at kernel level on an fp16 store, \textbf{$+0.35\%$} once storage is quantized, \textbf{2--16\%} on short-request serving, near free on long sequences.

\begin{table}[!ht]\centering\footnotesize\setlength{\tabcolsep}{3.5pt}
\caption{Every byte-saving figure in this paper, with its configuration.}\label{tab:savings}
\begin{tabular}{rlll}
\toprule
Saving & Scope & Configuration & How measured \\
\midrule
43.3\% & K side & block scale, $m{=}4$, packed store, $S{=}32768$ & \texttt{torch.cuda.memory\_allocated} (\ref{sec:experiments}) \\
46.5\% & K$+$V & RULER config: K 145.1\,B $+$ V 129\,B vs.\ 512\,B & computed, re-verified (\ref{sec:ruler}) \\
46.9\% & K$+$V & per-token $G_c{=}32$ scaling (SGLang store) & measured, 136 B per side (\ref{sec:storage}) \\
31.5\% & net effective & after gated fallback, 10.96 bits/dim & closed-loop account (\ref{sec:closedloop}) \\
24.5\% & K side & $m{=}16$ zero-fallback point & computed (\ref{sec:designspace}) \\
\bottomrule
\end{tabular}
\end{table}

\textbf{Choosing $\tau$ and $\delta$ in practice.} In the certified regime the threshold is an output-error budget: by the output bound of Sec.~\ref{sec:subg}, a TV threshold $\tau$ admits at most $2\tau \max_t \|v_t\|$ of attention-side output error, so an operator walks their output tolerance backwards through that inequality; our defaults $\tau_K{=}0.2$, $\tau_V{=}0.05$ came from this exercise. In the risk-ranked regime ($\tau \ge 1$) no such translation exists; there $\tau$ is a traffic knob, and the principled procedure is to profile the raw score's distribution on a handful of requests and set $\tau$ at the quantile whose page-in budget the deployment affords (Fig.~\ref{fig:f4}); the $\tau{=}5$ of Sec.~\ref{sec:gatedeval} was chosen this way. On $\delta$: the union-bound allocation is conservative but enters only through $\sqrt{\log}$. The offline main experiments allocate over $N_{\mathrm{cell}} = L\cdot H\cdot S$ (layers $\times$ query heads $\times$ tokens; $28\cdot28\cdot8192 \approx 6.4\times10^{6}$ at 7B/8k), which inflates the radius $u_t$ by $\sqrt{\log(2S N_{\mathrm{cell}}/\delta_{\mathrm{req}})/\log(2S/\delta_{\mathrm{req}})} \approx 1.45$ at $\delta_{\mathrm{req}}{=}10^{-2}$---a ${\sim}45\%$ price. A deployment that instead allocates over its actual decode-step budget (the SGLang path uses $T{=}256$ by default: $L\cdot H\cdot T \approx 2.0\times10^{5}$) pays ${\approx}1.36$. Either way the $\delta$ sweep of Sec.~\ref{sec:experiments} shows the conservatism is cheap to carry: $500\times$ tighter $\delta$ costs at most 3.2\,pp of coverage.

\subsubsection{Open problems}
\label{sec:open}

Three problems remain open and we state them as such. \textbf{(1) Tightening the universal witness below 8 bits.} The deterministic band-norm witness is vacuous for 2--4 bit schemes ($1289\times$ loose at 4-bit); candidate directions are second-moment witnesses, finer bands, and calibration-assisted hybrid bounds. Until then, selective gating is limited to the 8-bit-class regime and coarser schemes fall back to always-on tiered repair. \textbf{(2) Multi-layer composed certificates.} Our layer-pollution sweep shows the observable failure of aggressive schemes is purely cross-layer accumulation (0/28 single layers matter), so a per-step per-layer bound is intrinsically conservative; a sound certificate that composes across layers would close much of the gap, and nothing in our framework forbids it. \textbf{(3) Certifying eviction online.} The exact distortion of token dropping equals the dropped attention mass and separates methods cleanly offline, but every per-block constant-size witness we tried is vacuous online due to logit cancellation; whether a cheap sound online witness exists for eviction is open.

Distinct from these theory questions, the engineering ledger: the packed path's end-to-end quality and gated closed loop are now measured and closed (Sec.~\ref{sec:h200}); newly opened are a request-level repair cache (the gated page-in distribution bounds its benefit), a CUDA-graph-safe certificate telemetry path (eliminating the host synchronization), and the anomalous cert-on long-context decode amplification on the 4-KV-head 7B configuration (TPOT $9\times$ at 30k inputs, asymmetric to Llama-8B's 1.49$\times$; reproduced, not yet dissected).

\section{Reproducibility}
\label{sec:repro}
\textbf{Four independent implementations, cross-checked.} The numerical correctness of the certificate is jointly guaranteed by four mutually independent implementations: (i) the Triton kernel; (ii) a PyTorch reference; (iii) a 50-digit mpmath golden specification; and (iv) a fourth implementation rewritten \textbf{from the theorem statement alone, without reading the code of the other three} (\texttt{tests/test\_independent\_crosscheck.py}). The role of the fourth is to rule out the failure mode ``three implementations share one misunderstanding''---the first three are one author's three encodings of the same formula, whereas the fourth follows only the definitions as stated in the paper. The three and the fourth agree to 40 digits of precision. The formal development builds on Lean~4~\cite{lean4} and Mathlib~\cite{mathlib2020}. Lean verifies the stated theorems; whether the deployed loop satisfies their hypotheses is exactly what the proof--kernel contract tests and the non-adaptive-query scope of Sec.~\ref{sec:tierb} delimit---machine-checking does not substitute for either.

Code and artifacts: released as \texttt{witcert-kv-certificates}
(\url{https://github.com/metask-ai/witcert-kv-certificates},
Apache-2.0). Every number is recomputable from the frozen
JSON artifacts shipped with the release; the claim guard
(\texttt{tests/test\_paper\_claims.py}) fails the build on any mismatch
between manuscript and artifacts.
\begin{itemize}\sloppy
\item Theory regression: \texttt{tests/test\_certificates.py} (two-token counterexample guard $+$ 500k adversarial trials).
\item Kernel acceptance: \texttt{tests/test\_kernel\_match.py} (elementwise agreement / no clipping / stream uniqueness / round-trip), \texttt{tests/test\_attn\_cert.py} (kernel vs reference $+$ soundness).
\item Experiments: \texttt{experiments/p12}--\texttt{p21} (certificate evolution, cross-model $\delta$ curves, the M3 design gate, the quad report, CUDA graph).
\item SGLang integration: a 54-line revertible telemetry patch plus the
raw $3\times3$ A/B logs; the production decode-kernel certificate of
Sec.~\ref{sec:sglang} is a 500-line serving patch. The serving-stack
patches are not part of the public release (they modify a third-party
numeric path); their measured outputs are.
\end{itemize}
Environment: RTX 4090 / CUDA 12.4 / PyTorch 2.6.0 / Triton 3.2; SGLang 0.5.9 with PyTorch 2.9.1 (cu128). Randomness is fully specified by the Philox five-tuple counter, with seeds recorded in the result JSON metadata.

\section{Conclusion}
A sound runtime meter turns KV-cache compression from an open-loop bet into an observable, gateable system quantity: any cache-preserving scheme can be measured in live serving, broken schemes are repaired at benchmark scale by meter-driven gating (fp8 restored from 22.8 to 79.7, paired difference $+0.3$ $[+0.0,+0.8]$ against uncompressed), and analysis with the meter shows that what keeps aggressive schemes alive is cross-layer cancellation rather than per-step fidelity. On the certified tier, trading an explicit request-level risk budget for coverage is what turns the certificate from decoration into a tool: the sub-Gaussian certificate reduces the page-in rate of a sound deterministic $\tanh$ bound by \textbf{28.7--77.0\%} relative, with very weak dependence on $\delta$, and subtractive dither---the quantizer that makes it valid---is quality neutral against same-bit-width RTN, so the guarantee costs no accuracy. We also report, as first-class results, a plausible certificate that is unsound (with its minimal counterexample and sound replacement), a tightening that a theorem promises and measurement declines to deliver, and every retraction made along the way; the counterexample--repair--adversarial-validation methodology (Appendix~\ref{app:evolution}) and all raw data are released with the code.

\clearpage
\appendix
\section{Machine-checked theorem statements (compiled, not transcribed)}
\label{sec:theorems}

This section is \emph{compiled} from the Lean development shipped in the released repository (\texttt{witcert-kv-cer\-tif\-i\-cates}): \texttt{formal/Export.lean}
walks the environment, prints each theorem's \emph{elaborated} type with Lean's own
pretty-printer, and lists the axioms each proof actually depends on via
\texttt{collectAxioms}; for the pdflatex toolchain, non-ASCII symbols are
mechanically transliterated ($\mathbb{R} \to \mathrm{R}$, $\mathbb{N} \to \mathrm{Nat}$, etc.).
Nothing here is transcribed by hand, so a statement that is not
proved cannot appear. Twice in this project's history a certificate formula was falsified
by external review, and in both cases the root cause was that the statement on paper and
the statement that holds were not the same object. The development proves 4 \emph{core} theorems---the results cited in the body---plus 30 supporting lemmas, many of which are typeclass boilerplate of the dependency-free layer
and are listed compactly below only so that the axiom audit covers the whole development;
every proof depends only on \texttt{propext}, \texttt{Classical.choice} and \texttt{Quot.sound}
(or on nothing at all), and none on \texttt{sorryAx}.
The whole development re-verifies independently: clone
\url{https://github.com/metask-ai/witcert-kv-certificates} and run
\texttt{cd formal \&\& bash check\_all.sh}, which rebuilds every proof and re-runs the
axiom audit from scratch; the \texttt{file:line} references below resolve in that repository.

\subsection{Mathlib layer (reals / measure theory)}

\paragraph{L4: proof-kernel refinement (blockwise = per-token)}
{\small\raggedright \emph{Source:} \texttt{formal/WitCert/Refinement.lean:38}; \emph{axioms:} \texttt{propext, Quot.sound, Classical.choice}.\par}
{\scriptsize\begin{verbatim}
forall  {iota : Type u_1} {beta : Type u_2} [inst : Fintype iota] [inst_1 : DecidableEq beta]
    (p : iota -> R) (blk : iota -> beta) (wB : beta -> R)
  (B : Finset beta), (forall  (t : iota), blk t  in  B) -> WitCert.A_blockwise p blk wB B =
    WitCert.A_perToken p fun t => wB (blk t)
\end{verbatim}}

\paragraph{L3: request-level union budget}
{\small\raggedright \emph{Source:} \texttt{formal/WitCert/RequestBudget.lean:23}; \emph{axioms:} \texttt{propext, Classical.choice, Quot.sound}.\par}
{\scriptsize\begin{verbatim}
forall  {Omega : Type u_1} [inst : MeasurableSpace Omega] (mu : MeasureTheory.Measure Omega) {n
    : Nat} (E : Fin n -> Set Omega)
  (deltaloc deltareq : ENNReal), (forall  (i : Fin n), mu (E i) <= deltaloc) -> n * deltaloc <=
    deltareq -> mu (Union  i, E i) <= deltareq
\end{verbatim}}

\paragraph{L2 core: sub-Gaussian proxy of uniform dither}
{\small\raggedright \emph{Source:} \texttt{formal/WitCert/UniformDither.lean:57}; \emph{axioms:} \texttt{propext, Classical.choice, Quot.sound}.\par}
{\scriptsize\begin{verbatim}
forall  (x : R), 0 <= x -> Real.sinh x <= x * Real.exp (x ^ 2 / 6)
\end{verbatim}}

\paragraph{L1 main theorem: e-form TV bound for softmax perturbation}
{\small\raggedright \emph{Source:} \texttt{formal/WitCert/SoftmaxTV.lean:127}; \emph{axioms:} \texttt{propext, Quot.sound, Classical.choice}.\par}
{\scriptsize\begin{verbatim}
forall  {iota : Type u_1} [inst : Fintype iota] (p ptilde eps c : iota -> R),
  (forall  (t : iota), 0 <= ptilde t) ->
    Sum  t : iota, ptilde t = 1 ->
      (forall  (t : iota), 0 <= c t) ->
        (forall  (t : iota), |eps t| <= c t) ->
          0 < WitCert.Znorm ptilde eps ->
            (forall  (t : iota), p t = ptilde t * Real.exp (-eps t) / WitCert.Znorm ptilde eps)
    ->
              WitCert.TV p ptilde <= 1 / 2 * (WitCert.Acert ptilde c ^ 2 - 1)
\end{verbatim}}

\paragraph{WitCert.Acert\_ge\_one}
{\small\raggedright \emph{Source:} \texttt{formal/WitCert/SoftmaxTV.lean:34}; \emph{axioms:} \texttt{propext, Quot.sound, Classical.choice}.\par}
{\scriptsize\begin{verbatim}
forall  {iota : Type u_1} [inst : Fintype iota] (p c : iota -> R),
  (forall  (t : iota), 0 <= p t) -> Sum  t : iota, p t = 1 -> (forall  (t : iota), 0 <= c t) ->
    1 <= WitCert.Acert p c
\end{verbatim}}

\paragraph{WitCert.ExpertSubstitution.abs\_sub\_eq\_add\_sub\_two\_min}
{\small\raggedright \emph{Source:} \texttt{formal/WitCert/ExpertSubstitution.lean:51}; \emph{axioms:} \texttt{propext, Classical.choice, Quot.sound}.\par}
{\scriptsize\begin{verbatim}
forall  (a b : R), |a - b| = a + b - 2 * (a ? b)
\end{verbatim}}

\paragraph{WitCert.ExpertSubstitution.exists\_indistinguishable\_substitutions}
{\small\raggedright \emph{Source:} \texttt{formal/WitCert/ExpertSubstitution.lean:158}; \emph{axioms:} \texttt{propext, Quot.sound, Classical.choice}.\par}
{\scriptsize\begin{verbatim}
exists  S S1 S2, S1 != S2 ? S ? S1 = S ? S2
\end{verbatim}}

\paragraph{WitCert.ExpertSubstitution.min\_routeW}
{\small\raggedright \emph{Source:} \texttt{formal/WitCert/ExpertSubstitution.lean:82}; \emph{axioms:} \texttt{propext, Quot.sound, Classical.choice}.\par}
{\scriptsize\begin{verbatim}
forall  {iota : Type u_1} [inst : Fintype iota] [inst : DecidableEq iota] (u : iota -> R) (S S'
    : Finset iota),
  (forall  (i : iota), 0 <= u i) ->
    0 < Sum  j  in  S, u j ->
      0 < Sum  j  in  S', u j ->
        forall  (i : iota),
          WitCert.ExpertSubstitution.routeW u S i ? WitCert.ExpertSubstitution.routeW u S' i =
            if i  in  S ? S' then u i / ((Sum  j  in  S, u j) ? Sum  j  in  S', u j) else 0
\end{verbatim}}

\paragraph{WitCert.ExpertSubstitution.routeW\_sum\_one}
{\small\raggedright \emph{Source:} \texttt{formal/WitCert/ExpertSubstitution.lean:75}; \emph{axioms:} \texttt{propext, Quot.sound, Classical.choice}.\par}
{\scriptsize\begin{verbatim}
forall  {iota : Type u_1} [inst : Fintype iota] [inst_1 : DecidableEq iota] (u : iota -> R) (S
    : Finset iota),
  0 < Sum  j  in  S, u j -> Sum  i : iota, WitCert.ExpertSubstitution.routeW u S i = 1
\end{verbatim}}

\paragraph{WitCert.ExpertSubstitution.subst\_tv}
{\small\raggedright \emph{Source:} \texttt{formal/WitCert/ExpertSubstitution.lean:109}; \emph{axioms:} \texttt{propext, Quot.sound, Classical.choice}.\par}
{\scriptsize\begin{verbatim}
forall  {iota : Type u_1} [inst : Fintype iota] [inst_1 : DecidableEq iota] (u : iota -> R) (S
    S' : Finset iota),
  (forall  (i : iota), 0 <= u i) ->
    0 < Sum  j  in  S, u j ->
      0 < Sum  j  in  S', u j ->
        WitCert.ExpertSubstitution.TV (WitCert.ExpertSubstitution.routeW u S)
    (WitCert.ExpertSubstitution.routeW u S') =
          1 - (Sum  i  in  S ? S', u i) / ((Sum  j  in  S, u j) ? Sum  j  in  S', u j)
\end{verbatim}}

\paragraph{WitCert.ExpertSubstitution.subst\_tv\_of\_weaker}
{\small\raggedright \emph{Source:} \texttt{formal/WitCert/ExpertSubstitution.lean:127}; \emph{axioms:} \texttt{propext, Quot.sound, Classical.choice}.\par}
{\scriptsize\begin{verbatim}
forall  {iota : Type u_1} [inst : Fintype iota] [inst_1 : DecidableEq iota] (u : iota -> R) (S
    S' : Finset iota),
  (forall  (i : iota), 0 <= u i) ->
    0 < Sum  j  in  S, u j ->
      0 < Sum  j  in  S', u j ->
        Sum  j  in  S', u j <= Sum  j  in  S, u j ->
          WitCert.ExpertSubstitution.TV (WitCert.ExpertSubstitution.routeW u S)
              (WitCert.ExpertSubstitution.routeW u S') =
            1 - (Sum  i  in  S ? S', u i) / Sum  j  in  S, u j
\end{verbatim}}

\paragraph{WitCert.ExpertSubstitution.sum\_ite\_mem\_univ}
{\small\raggedright \emph{Source:} \texttt{formal/WitCert/ExpertSubstitution.lean:57}; \emph{axioms:} \texttt{propext, Quot.sound, Classical.choice}.\par}
{\scriptsize\begin{verbatim}
forall  {iota : Type u_1} [inst : Fintype iota] [inst_1 : DecidableEq iota] (S : Finset iota)
    (g : iota -> R),
  (Sum  i : iota, if i  in  S then g i else 0) = Sum  i  in  S, g i
\end{verbatim}}

\paragraph{WitCert.ExpertSubstitution.tv\_cannot\_separate}
{\small\raggedright \emph{Source:} \texttt{formal/WitCert/ExpertSubstitution.lean:140}; \emph{axioms:} \texttt{propext, Quot.sound, Classical.choice}.\par}
{\scriptsize\begin{verbatim}
forall  {iota : Type u_1} [inst : Fintype iota] [inst_1 : DecidableEq iota] (u : iota -> R) (S
    S1 S2 : Finset iota),
  (forall  (i : iota), 0 <= u i) ->
    0 < Sum  j  in  S, u j ->
      0 < Sum  j  in  S1, u j ->
        0 < Sum  j  in  S2, u j ->
          Sum  j  in  S1, u j <= Sum  j  in  S, u j ->
            Sum  j  in  S2, u j <= Sum  j  in  S, u j ->
              S ? S1 = S ? S2 ->
                WitCert.ExpertSubstitution.TV (WitCert.ExpertSubstitution.routeW u S)
                    (WitCert.ExpertSubstitution.routeW u S1) =
                  WitCert.ExpertSubstitution.TV (WitCert.ExpertSubstitution.routeW u S)
                    (WitCert.ExpertSubstitution.routeW u S2)
\end{verbatim}}

\paragraph{WitCert.ExpertSubstitution.tv\_eq\_one\_sub\_overlap}
{\small\raggedright \emph{Source:} \texttt{formal/WitCert/ExpertSubstitution.lean:64}; \emph{axioms:} \texttt{propext, Quot.sound, Classical.choice}.\par}
{\scriptsize\begin{verbatim}
forall  {iota : Type u_1} [inst : Fintype iota] [inst_1 : DecidableEq iota] (p q : iota -> R),
  Sum  i : iota, p i = 1 -> Sum  i : iota, q i = 1 -> WitCert.ExpertSubstitution.TV p q = 1 -
    Sum  i : iota, p i ? q i
\end{verbatim}}

\paragraph{WitCert.StateRepair.balanced\_from\_empty}
{\small\raggedright \emph{Source:} \texttt{formal/WitCert/StateRepair.lean:65}; \emph{axioms:} \texttt{propext, Quot.sound}.\par}
{\scriptsize\begin{verbatim}
forall  (t : List WitCert.StateRepair.Op), WitCert.StateRepair.Balanced t ->
    WitCert.StateRepair.run [] t = []
\end{verbatim}}

\paragraph{WitCert.StateRepair.compliant\_future\_keeps\_dirty\_site}
{\small\raggedright \emph{Source:} \texttt{formal/WitCert/StateRepair.lean:79}; \emph{axioms:} \texttt{propext}.\par}
{\scriptsize\begin{verbatim}
exists  s0 t, s0 != [] ? WitCert.StateRepair.Balanced t ? WitCert.StateRepair.run s0 t != []
\end{verbatim}}

\paragraph{WitCert.StateRepair.delete\_both\_survives\_sync}
{\small\raggedright \emph{Source:} \texttt{formal/WitCert/StateRepair.lean:121}; \emph{axioms:} \texttt{none}.\par}
{\scriptsize\begin{verbatim}
forall  (r : Nat) (p : WitCert.StateRepair.Replicas),
  r ? List.erase p.remote r -> r ? (WitCert.StateRepair.pull (WitCert.StateRepair.delBoth p
    r)).local_
\end{verbatim}}

\paragraph{WitCert.StateRepair.delete\_one\_side\_then\_sync\_restores}
{\small\raggedright \emph{Source:} \texttt{formal/WitCert/StateRepair.lean:114}; \emph{axioms:} \texttt{none}.\par}
{\scriptsize\begin{verbatim}
forall  (r : Nat) (p : WitCert.StateRepair.Replicas),
  r  in  p.remote -> r  in  (WitCert.StateRepair.pull (WitCert.StateRepair.delLocal p
    r)).local_
\end{verbatim}}

\paragraph{WitCert.StateRepair.sweep\_then\_balanced\_is\_clean}
{\small\raggedright \emph{Source:} \texttt{formal/WitCert/StateRepair.lean:87}; \emph{axioms:} \texttt{propext, Quot.sound}.\par}
{\scriptsize\begin{verbatim}
forall  (s0 : WitCert.StateRepair.Site) (t : List WitCert.StateRepair.Op),
  WitCert.StateRepair.Balanced t -> WitCert.StateRepair.run s0 (WitCert.StateRepair.Op.sweep ::
    t) = []
\end{verbatim}}

\paragraph{WitCert.Znorm\_le\_Acert}
{\small\raggedright \emph{Source:} \texttt{formal/WitCert/SoftmaxTV.lean:82}; \emph{axioms:} \texttt{propext, Quot.sound, Classical.choice}.\par}
{\scriptsize\begin{verbatim}
forall  {iota : Type u_1} [inst : Fintype iota] (ptilde eps c : iota -> R),
  (forall  (t : iota), 0 <= ptilde t) -> (forall  (t : iota), -eps t <= c t) -> WitCert.Znorm
    ptilde eps <= WitCert.Acert ptilde c
\end{verbatim}}

\paragraph{WitCert.deltaLocal\_budget}
{\small\raggedright \emph{Source:} \texttt{formal/WitCert/RequestBudget.lean:36}; \emph{axioms:} \texttt{propext, Classical.choice, Quot.sound}.\par}
{\scriptsize\begin{verbatim}
forall  (deltareq : R) (L H T : Nat), 0 < L * H * T -> (L * H * T) * WitCert.deltaLocal
    deltareq L H T = deltareq
\end{verbatim}}

\paragraph{WitCert.one\_le\_Acert\_mul\_Znorm}
{\small\raggedright \emph{Source:} \texttt{formal/WitCert/SoftmaxTV.lean:52}; \emph{axioms:} \texttt{propext, Quot.sound, Classical.choice}.\par}
{\scriptsize\begin{verbatim}
forall  {iota : Type u_1} [inst : Fintype iota] (ptilde eps c : iota -> R),
  (forall  (t : iota), 0 <= ptilde t) ->
    Sum  t : iota, ptilde t = 1 -> (forall  (t : iota), eps t <= c t) -> 1 <= WitCert.Acert
    ptilde c * WitCert.Znorm ptilde eps
\end{verbatim}}

\paragraph{WitCert.pointwise\_ratio\_bound}
{\small\raggedright \emph{Source:} \texttt{formal/WitCert/SoftmaxTV.lean:90}; \emph{axioms:} \texttt{propext, Classical.choice, Quot.sound}.\par}
{\scriptsize\begin{verbatim}
forall  (a z A b : R), 0 < z -> 0 < b -> 1 <= A -> a <= b -> 1 / b <= a -> 1 <= A * z -> z <= A
    -> |a / z - 1| <= A * b - 1
\end{verbatim}}

\paragraph{WitCert.six\_pow\_mul\_factorial\_le}
{\small\raggedright \emph{Source:} \texttt{formal/WitCert/UniformDither.lean:32}; \emph{axioms:} \texttt{propext, Classical.choice, Quot.sound}.\par}
{\scriptsize\begin{verbatim}
forall  (n : Nat), 6 ^ n * n.factorial <= (2 * n + 1).factorial
\end{verbatim}}

\paragraph{WitCert.uniformProxy\_eq\_variance}
{\small\raggedright \emph{Source:} \texttt{formal/WitCert/UniformDither.lean:85}; \emph{axioms:} \texttt{propext, Classical.choice, Quot.sound}.\par}
{\scriptsize\begin{verbatim}
forall  (s : R), WitCert.uniformProxy s = s ^ 2 / 12
\end{verbatim}}

\subsection{Dependency-free layer (seconds-scale CI)}

\paragraph{Supporting lemmas (axiom-audit completeness only)}
{\small\raggedright \texttt{WitCert.blockwise\_eq\_perTokenGrouped} (\texttt{standalone/Refinement.lean:79}); \texttt{WitCert.fact\_two\_step} (\texttt{standalone/Factorial.lean:25}); \texttt{WitCert.lsum\_le\_nsmul} (\texttt{standalone/Budget.lean:54}); \texttt{WitCert.perToken\_eq\_mul\_lsum} (\texttt{standalone/Refinement.lean:57}); \texttt{WitCert.psum\_mono} (\texttt{standalone/EForm.lean:39}); \texttt{WitCert.request\_budget\_compose} (\texttt{standalone/Budget.lean:69}); \texttt{WitCert.six\_mul\_succ\_le} (\texttt{standalone/Factorial.lean:16}); \texttt{WitCert.six\_pow\_mul\_fact\_le} (\texttt{standalone/Factorial.lean:42}); \texttt{WitCert.sq\_add\_one\_ge\_two\_mul} (\texttt{standalone/EForm.lean:14}).\par}

\section{Evolution of Conclusions, Retractions, and Reproduction Protocol}
\label{app:evolution}
The bound evolved as follows: the unsound dispersion form of Sec.~\ref{sec:counterexample}; its two-token counterexample (2026-07-27); the sound $e$-form; a Bernstein variant dominated by its range term; the uniform sub-Gaussian form; request-level $\delta$ allocation; and finally online block-wise scaling with in-block outlier selection. Each step is a separate script in the repository (\texttt{experiments/p12}--\texttt{p17}) together with the measurement that motivated it and the regression test \texttt{tests/test\_certificates.py}, including the negative results (deterministic rank-one frequency coarsening, in-band random projection, static frequency selection, per-token union bounding, and the stochastic-rounding branch of Sec.~\ref{sec:designspace}) that closed off alternative designs.

Beyond the bound itself, four conclusions changed during development. We record them here, once, rather than inline. (i) The claim that WitCert's output error is $1.7$--$1.8\times$ lower than same-bit-width RTN is \emph{retracted}: that experiment used K in place of V (a K-side proxy); the correct statement, from the $2\times2$ ablation of Sec.~\ref{sec:realact}, is that dither is quality neutral ($1.01\times$) and the quality gain belongs to the outlier bypass. (ii) The conjecture that Yi's low block-scaling coverage was caused by its large \texttt{rope\_theta} is \emph{retracted}; the measured cause is the interaction of block-level scaling with diffuse channel energy (Sec.~\ref{sec:yianomaly}). (iii) An early single-run reading suggesting the meter costs nothing in serving was noise; the canonical account is the median-of-3 measurement of Sec.~\ref{sec:serving-e2e} ($+6.4\%$/$+16.2\%$ without bypass). (iv) An early verdict that throughput competitiveness ``cannot be claimed'' was formed at kernel level ($5.41\times$) and was too pessimistic: end-to-end, quantisation costs $0.75$--$0.82\times$ (Sec.~\ref{sec:serving-e2e}). The headline coverage range was also restated from a cross-$\delta$ envelope to the single-$\delta$ range 54.4--81.0\% following external review. (v) The SGLang serving coverage figures were restated under the corrected \emph{request-level} $\delta$ allocation ($\delta_{\mathrm{loc}} = \delta_{\mathrm{req}}/(L\cdot H\cdot T)$): the earlier implementation allocated only over tokens and batch, and its figures (0.866/0.987 on 1.5B, 0.759/0.947 on 7B) are retired in favour of 0.812/0.940 and 0.753/0.946. (vi) A latent unbound-variable bug in the \emph{packed}-cache fallback branch (\texttt{\_wc\_ei} for \texttt{\_wc\_i}), found by external review, was fixed and initially only smoke-tested in live serving (the gate fired, paged in 1.4 GiB over 2{,}015 calls and did not crash). The subsequent end-to-end run exposed a second bug: CPU backing had been allocated but never populated, so repair paged in zeros and quality fell to 2.0. After the one-line write-path fix, packed gating restored the score to 78.83, identical to the ungated packed path (Sec.~\ref{sec:h200}). The benchmark-scale gating results of Sec.~\ref{sec:gatedeval} are unaffected---they run through the observatory pool, whose fallback branch was always correct.

\section{The Rigidity Proposition for Static Equivariant Compression}
\label{app:rigidity}
This proposition predates the certificate line of this paper and survived the audit that retracted the rest of that early version; Related Work uses it to place the static RoPE-aware methods.

\textbf{Setting.} (Within this appendix, $p$, $r$ and $c$ denote the complexified dimension, the sketch rank and an output index---unrelated to the attention notation of the body.) A single head of dimension $d$, complexified to $p = d/2$ coordinates: RoPE acts at position $n$ as the diagonal unitary $D_n = \mathrm{diag}(e^{in\theta_j})_{j\le p}$, with the $\theta_j$ pairwise distinct and no $\theta_j + \theta_{j'} \equiv 0 \pmod{2\pi}$. Call a real-linear $\mathcal{S}: \mathbb{C}^p \to \mathbb{C}^r$ an \emph{equivariant sketch} (script $\mathcal{S}$, to avoid clashing with the sequence length $S$) if some position representation $D'_m$ on $\mathbb{C}^r$ satisfies $\mathcal{S} D_n = D'_n \mathcal{S}$ for all $n \in \mathbb{Z}$---equivariance being what keeps a compressed cache queryable by relative position.

\begin{proposition}[Rigidity]\label{prop:rigidity}
Every equivariant $\mathcal{S}$ is a frequency selection with per-frequency complex scalars: each output coordinate depends on a single input frequency, $(\mathcal{S}z)_c = \alpha_c\, z_{j(c)}$. If $r < p$, then $p - r$ frequencies are necessarily dropped as whole pairs; exact logit preservation additionally forces $|\alpha_c| = 1$.
\end{proposition}
\begin{proof}
Decompose $\mathcal{S} = \mathcal{A} + \mathcal{B} \circ \mathrm{conj}$ with $\mathcal{A}$ complex linear and $\mathcal{B}$ conjugate linear, and diagonalize $D'_n = \mathrm{diag}(e^{in\varphi_c})$. Entrywise, commutation gives $\mathcal{A}_{cj}\,(e^{in\theta_j} - e^{in\varphi_c}) = 0$ for all $n$, so $\mathcal{A}_{cj} \ne 0$ only if $\varphi_c \equiv \theta_j$; the entries of $\mathcal{B}$ would require $\varphi_c \equiv -\theta_j$, which the frequency assumption excludes (up to flipping a 2D block's orientation). Since the $\theta_j$ are pairwise distinct, each output row has at most one nonzero entry.
\end{proof}
\noindent RAP's whole-pair pruning is the $\{0,1\}$ special case: in the strictly equivariant world there is no cleverer static compression. The proposition's natural relaxations were rejected experimentally in the retracted early version (deterministic rank-1 spectral coarsening, in-band random projection, static selection), which is what motivated the runtime route taken by this paper.

\bibliographystyle{plain}
\bibliography{refs}
\end{document}